\documentclass[a4paper,12pt,final,reqno]{amsart}

\usepackage{amsmath,amssymb,amsthm}

\numberwithin{equation}{section}
\allowdisplaybreaks

\theoremstyle{definition}
 \newtheorem{thm}{Theorem}[section]
 \newtheorem{prp}[thm]{Proposition}
 \newtheorem{lem}[thm]{Lemma}
 \newtheorem{cor}[thm]{Corollary}
 
 \newtheorem{fct}[thm]{Fact}
 \newtheorem{dfn}[thm]{Definition}

 \newtheorem{rmk}[thm]{Remark}

\newcommand{\ket}[1]{\left|#1 \right>}
\newcommand{\bra}[1]{\left< #1 \right|}
\newcommand{\seteq}{\mathbin{:=}}

\newcommand{\simto}{\xrightarrow{\,\sim\,}}

\newcommand{\bbR}{\mathbb{R}}

\newcommand{\bbC}{\mathbb{C}}
\newcommand{\bbF}{\mathbb{F}}
\newcommand{\bbN}{\mathbb{N}}
\newcommand{\bbP}{\mathbb{P}}
\newcommand{\bbQ}{\mathbb{Q}}

\newcommand{\bbZ}{\mathbb{Z}}

\newcommand{\bfv}{\mathbf{v}}

\newcommand{\calF}{\mathcal{F}}
\newcommand{\calH}{\mathcal{H}}

\newcommand{\calO}{\mathcal{O}}

\newcommand{\calU}{\mathcal{U}}

\newcommand{\FF}{\widetilde{\bbF}}

\newcommand{\tPhi}{{\widetilde{\Phi}}}
\newcommand{\tPhis}{{\widetilde{\Phi}^*}}

\title{Quantum Algebraic Approach to 
Refined Topological Vertex}
\author{H.~Awata, B.~Feigin and J.~Shiraishi}
\address{HA: Graduate School of Mathematics, Nagoya University, Nagoya, 464-8602, Japan}
\email{awata@math.nagoya-u.ac.jp}
\address{BF: Landau Institute for Theoretical Physics,
Russia, Chernogolovka, 142432, prosp. Akademika Semenova, 1a,   
\\
Higher School of Economics, Russia, Moscow, 101000,  Myasnitskaya ul., 20, and
\\
Independent University of Moscow, Russia, Moscow, 119002,
Bol'shoi Vlas'evski per., 11}
\email{borfeigin@gmail.com}
\address{MK,JS: Graduate School of Mathematical Sciences, University of Tokyo, Komaba, Tokyo 153-8914, Japan}
\email{shiraish@ms.u-tokyo.ac.jp}

\begin{document}

\begin{abstract}
We establish the equivalence between the 
refined topological vertex of Iqbal-Kozcaz-Vafa and 
a certain representation theory  of the quantum algebra of type $W_{1+\infty}$ introduced by Miki.
Our construction involves trivalent intertwining operators $\Phi$ and $\Phi^*$ associated with 
triples of the bosonic Fock modules. Resembling the 
topological vertex,  a triple of 
vectors $\in \bbZ^2$ is attached to each intertwining operator, which satisfy the
Calabi-Yau and smoothness conditions. 
It is shown that 
certain matrix elements of $\Phi$ and $\Phi^*$ give the 
refined topological vertex $C_{\lambda\mu\nu}(t,q)$ of Iqbal-Kozcaz-Vafa.
With another choice of basis,  we recover the refined topological vertex ${C_{\lambda\mu}}^\nu(q,t)$ of Awata-Kanno.  
The gluing factors appears correctly when we consider any compositions 
of $\Phi$ and $\Phi^*$.
The spectral parameters attached to Fock spaces play the role of the K\"ahler parameters.
\end{abstract}

\maketitle

\section{Introduction}
The aim of the present paper is to study
the refined topological vertex $C_{\lambda\mu\nu}(t,q)$ of Iqbal, Kozcaz and Vafa 
\cite{IKV:2009}
from the point of view of the quantum algebra of type $W_{1+\infty}$ introduced by Miki \cite{Mi:2007}. 
We also treat the vertex  ${C_{\lambda\mu}}^\nu(q,t)$ considered by Awata and Kanno \cite{AK:2009}
in the same footing.

Let us first recall briefly the notion of the topological vertex \cite{AKMV:2005}, \cite{I:2002}. 
A trivalent graph plays an important role, since
it encodes the information where the cycles of  a $T^2$ fibration of a toric 3-fold degenerate.
The Calabi-Yau threefold is then mapped to a Feynman graph with fixed Schwinger terms (K\"ahler classes
of the threefold), and 
the topological vertex is associated with states in the threefold tensor product 
of bosonic Fock spaces. Each edges of the graph is an oriented straight line
labeled by a vector ${\bf v}\in \bbZ^2$ corresponding to 
the generator of $H_1(T^2)$ (shrinking cycles). 
If all the edges are incoming, we have the condition 
$\sum_i{\bf v}_i=0$ (Calabi-Yau condition), and
$|{\bf v}_i\wedge {\bf v}_j|=1$ for any pair of edges (smoothness condition).
Together with a `gluing rules,' on can calculate all genus amplitudes of the topological
A-model for non-compact toric Calabi-Yau threefolds. 
The topological vertex $C_{\lambda\mu\nu}$ is represented by 
Okounkov, Reshetikhin and Vafa using the 
skew Schur functions  \cite{ORV:2003}, \cite{OR:2007}
\begin{align}
C_{\lambda\mu\nu}(q)=
 q^{\kappa(\mu)\over 2}
s_{\nu'}(q^{-\rho})
\sum_\eta 
s_{\lambda'/\eta}(q^{-\nu-\rho})s_{\mu/\eta}(q^{-\nu'-\rho}),
\end{align}
where $\lambda,\mu,\nu$ are partitions labeling the states in the 
threefold tensor of the Fock spaces, $\lambda'$ denotes the transpose of $\lambda$,
$\rho=(-1/2,-3/2,-5/2,\cdots)$,
and 
$\kappa(\lambda)=\sum_i \lambda_i(\lambda_i+1-2 i)$.

In \cite{IKV:2009} a refined version of the topological vertex was introduced, 
based on the arguments of geometric engineering  concerning 
the $K$-theoretic lift of the Nekrasov partition functions \cite{N:2003}, \cite{FP:2003}. 
See also \cite{NY:2005:a}, \cite{NY:2005:b}.
In this refined version, 
one more parameter $t$ comes in and the theory seems to be deeply 
relate with the Macdonald functions $P_\lambda(x;q,t)$ \cite{Ma:1995}.
The formula is 
\begin{align}
&
C^{{\rm (IKV)}}_{\lambda\mu\nu}(t,q)=
\left(q\over t\right)^{||\mu||^2\over 2} t^{\kappa(\mu)\over 2} q^{||\nu||^2\over 2}
\widetilde{Z}_\nu(t,q)
\sum_\eta 
\left(q\over t\right)^{|\eta|+|\lambda|-|\mu|\over 2} 
s_{\lambda'/\eta}(t^{-\rho}q^{-\nu})s_{\mu/\eta}(t^{-\nu'}q^{-\rho}),\\
&
\widetilde{Z}_\nu(t,q)=\prod_{s\in \nu} (1-q^{a_\nu(s)}t^{\ell_\nu(s)+1})^{-1}
= t^{-{||\nu'||^2\over 2}}P_\nu(t^{-\rho};q,t),
\end{align}
where $||\lambda||^2=\sum_i\lambda_i^2$.
See \cite{IK:2011} for recent development, and a remark on their notational convention.

There is another approach by Awata and Kanno \cite{AK:2009},
where Macdonald  functions are used in some symmetric way 
\begin{align}
&
{C_{\mu\lambda}}^\nu(q,t)=
P_\lambda(t^\rho;q,t)
\sum_\sigma \iota P_{\mu'/\sigma'}(-t^{\lambda'}q^{\rho};t,q)
P_{\nu/\sigma}(q^\lambda t^\rho;q,t) (q^{1/2}/t^{1/2})^{|\sigma|-|\nu|}
f_\nu(q,t)^{-1}. \label{AK}
\end{align}
(See Section 4.1  as for the notations.)
Here they incorporated the `framing factor' 
\begin{align}
f_\lambda(q,t)=
(-1)^{|\lambda|} q^{n(\lambda')+|\lambda|/2} t^{-n(\lambda)-|\lambda|/2} , \label{f-lam}
\end{align}
which was introduced by Taki \cite{T:2008}.
It has been recognized that these two different formulas give us essentially the same result, and the 
difference should be superficial.  
As for the preliminary version of the formula (\ref{AK}), see \cite{AK:2005}. 
\bigskip

Now we turn to the quantum algebra side.
The algebra we consider was first introduced by Miki  in 
his study on the $W_{1+\infty}$ algebra. 
After the first discovery by Miki, essentially the same algebraic structure has been rediscovered by several 
authors. See 
\cite{FT:2009}, \cite{FHHSY:2009}, \cite{SV1:2009}, \cite{SV2:2009}, \cite{Sc1:2010},
\cite{Sc2:2010},
\cite{FFJMM1:2010}, \cite{FFJMM2:2010}, \cite{FJMM:2011}.
This verifies the naturalness 
and the richness of the algebra. Because of this, 
it has been called by several different names, and there is no good choice at this moment
than waiting for well established terminologies. In this paper, we denote the algebra by $\calU$.

Motivated by the construction in (refined) topological vertex,
we study a representation theory of the quantum algebra $\calU$ which 
includes the following ingredients:
\begin{enumerate}
\item triple of the Fock spaces and associated intertwining operators, 
\item trivalent vertex with edges labeled by vectors $\in \bbZ^2$ satisfying the 
Calabi-Yau and smoothness conditions,
\item spectral parameters playing the role of the K\"ahler parameters.
\end{enumerate}

It has been recognized that the quantum algebra $\calU$ has two central elements, and
they obey a certain transformation formula 
with respect to the $SL(2,\bbZ)$ action \cite{Mi:2007},
\cite{SV1:2009,SV2:2009}.
Namely, the $SL(2,\bbZ)$ action
preserves the structure of the algebra
up to the shift in the  central elements. 
As a consequence of the $SL(2,\bbZ)$ action, we have two types of the Fock representations of $\calU$, 
one in \cite{FT:2009} and the other in \cite{FHHSY:2009}.
After fixing convention suitably, 
one can say that the former has level $(0,1)$ (vertical), and 
the latter has level $(1,0)$ (horizontal).
The action of the $T$ generator of the $SL(2,\bbZ)$ can be 
easily treated, and we can modify the `horizontal' Fock representation
to level $(1,N)$ with $N\in \bbZ$. 
We restrict ourselves only to the family of the Fock modules
$\calF^{(0,1)}_u$ and $\calF^{(1,N)}_u$ ($N\in \bbZ$), where 
$u$ is the  spectral parameter. (See Sections 2.3, 2.4.)
Note if one of the edges (the preferred edge) is labeld by $(0,1)$,
then from the Calabi-Yau and the smoothness condition
the rest should be $(1,N)$ and $(-1,-N-1)$ where $N\in \bbZ$.

Consider the  intertwining operators of $\calU$-modules 
associated with three Fock modules of the forms
$\Phi=\Phi\left[{{\bf v}_3,u_3\atop {\bf v}_1,u_1;{\bf v}_2,u_2}\right]:
\calF^{{\bfv}_1}_{u_1}\otimes \calF^{{\bfv}_2}_{u_2}\rightarrow \calF^{{\bfv}_3}_{u_3}$
and 
$\Phi^*=\Phi^*
\left[{{\bf v}_2,u_2; {\bf v}_1,u_1\atop{\bf v}_3,u_3}\right]:
\calF^{{\bfv}_3}_{u_3}\rightarrow\calF^{{\bfv}_2}_{u_2}\otimes \calF^{{\bfv}_1}_{u_1}$.
The following particular cases are essential in our construction: 
\begin{align}
&
\Phi:\calF^{(0,1)}_v\otimes \calF^{(1,N)}_u\longrightarrow
\calF^{(1,N+1)}_{-vu},\qquad 
 a \Phi= \Phi \Delta(a) \qquad (\forall a\in \calU),\label{e-1}\\
&
 \Phi_\lambda (\alpha)=\Phi(P_\lambda \otimes \alpha )
\qquad (\forall  P_\lambda\otimes \alpha \in \calF^{(0,1)}_v\otimes \calF^{(1,N)}_u),
\qquad \Phi_\emptyset (1)=1+\cdots ,\label{e-2}\\
&
\Phi^*:
\calF^{(1,N+1)}_{-vu}\longrightarrow
\calF^{(1,N)}_v\otimes 
\calF^{(0,1)}_u,\qquad 
 \Delta(a) \Phi^*= \Phi^* a \qquad (\forall a\in \calU),\label{e-3}\\
 &
 \Phi^*( \alpha )=\sum_\lambda \Phi^*_\lambda (\alpha)\otimes Q_\lambda
\qquad (\forall \alpha \in \calF^{(1,N+1)}_{-vu}),
\qquad 
\Phi^*_\emptyset (1)=1+\cdots ,\label{e-4}
\end{align}
where $\Phi_\lambda$ and  $\Phi^*_\lambda$ are normalized components of 
$\Phi$ and $\Phi^*$.
We prove that such (normalized) intertwining operators $\Phi$ and $\Phi^*$ exist
uniquely. (Theorems \ref{thm-1}, \ref{thm-2}.)

Let $S_\lambda(q,t)$'s be the dual of the Schur function $s_\mu$'s with respect to
the Macdonald scalar product in (\ref{scalar}) satisfying 
$\langle S_\lambda(q,t),s_\mu\rangle_{q,t}=\delta_{\lambda,\mu}$.
We show that the refined topological vertex $C_{\lambda\mu\nu}^{\rm (KIV)}(t,q)$
coincides with the matrix element
$\bra{S_\mu(q,t)} \Phi^*_\nu \ket{s_{\lambda'}}$ ut to 
a simple factor. (Proposition \ref{mat-el-IKV}.)
If we use the bases $(\iota P_\lambda)$ and  $(\iota Q_\mu)$, then
 the refined topological vertex ${C_{\lambda\mu}}^\nu$ 
arises as the 
 matrix element
$\bra{\iota P_\nu} \Phi^*_\lambda \ket{\iota Q_\mu}$. (Proposition \ref{mat-el}.)

We cheek that any types of the compositions of the intertwining operators 
$\Phi$ and $\Phi^*$ produces contractions of topological vertices involving 
correct gluing factors (see Definition \ref{gluing-rules}). 
Thereby proving the equivalence of the topological vertex 
and our representation theory. (Theorem \ref{equiv}.)
Since the discovery of Alday, Gaiotto and Tachikawa  \cite{AGT:2010}, 
it has been intensively studied that 
we have the representation theory of the Virasoro and $W$ algebras 
playing a profound role in the Nekrasov instanton partition function
\cite{HJS:2010},
\cite{FL:2010},
\cite{MMS:2011},
\cite{AFLT:2011}. As for the $K$-theoretic version, see \cite{AY:2010}, \cite{AY:2010:2},
\cite{AFHKSY:2011}.
Our quantum algebraic approach extends this idea in such a way that the topological A-model
and the topological vertex are involved.
We hope this will give us better understandings both in string theory side and in 
quantum integrable system side.

We remark that the intertwining operator proposed in \cite{AFHKSY:2011}
can be explicitly constructed by composing our $\Phi$'s and $\Phi^*$'s in a suitable manner.
However it still remains unclear how to 
describe the the structure of the `integral basis' proposed there. 
Therefore we do not go in this direction here.
\bigskip

This paper is organized as follows.
In Section 2, we recall our notations for the algebra $\calU$ 
and introduce  the family of the Fock modules
$\calF^{(0,1)}_u$ and $\calF^{(1,N)}_u$ ($N\in \bbZ$).
In Section 3, the trivalent intertwining operators $\Phi$ and $\Phi^*$ are 
defined (Definitions \ref{def-phi} and \ref{def-phis}). The existence and the uniqueness of them are stated in 
Theorems \ref{thm-1} and \ref{thm-2}. 
Section 4 is devoted to establishing the equivalence with 
topological vertex of Iqbal-Kozcaz-Vafa, and Awata-Kanno.
For this purpose, we calculate the matrix elements of  $\Phi$ and $\Phi^*$ 
(Propositions \ref{mat-el-IKV} and \ref{mat-el}). 
Then we check the gluing rules for all the possible compositions (Propositions 
\ref{gl-pro-1}, \ref{gl-pro-2}, \ref{gl-pro-3}, \ref{gl-pro-4} and \ref{gl-pro-5}).
Our main theorem is stated in Theorem \ref{equiv}.
In Section 5, we present two examples of calculations which involve 
Nekrasov partition functions, and investigate the meaning of the spectral parameters.
Proofs of Theorems \ref{thm-1} and \ref{thm-2} are given in Section 6.
A proof of Proposition  \ref{propos-1} is stated in Section 7.

\section{Preliminaries}
\subsection{Algebra $\calU$}
Let $q$ and $t$ be independent indeterminates, and set 
$\bbF=\bbQ(q,t)$. Set also
$\FF=\bbQ(q^{1/4},t^{1/4})$. We sometimes work over the field 
$\FF$ to keep the notation considerably symmetric
for our dual constructions. 
We briefly recall our notation for the algebra $\calU$. 
We follow the notation in \cite{FHHSY:2009} which is based on \cite{DI:1997}.
Let
\begin{align}
g(z)=\dfrac{G^+(z)}{G^-(z)}\in \bbQ(q,t)[[z]],\qquad
G^\pm(z)=(1-q^{\pm1}z)(1-t^{\mp 1}z)(1-q^{\mp1}t^{\pm 1}z).\label{g}
\end{align}

\begin{dfn}\label{dfn:calU(q,t)}
Let $\calU$ be a unital associative algebra over $\bbF$ generated by
the Drinfeld currents 
$x^\pm(z)=\sum_{n\in \bbZ}x^\pm_n z^{-n},
\psi^\pm(z)=\sum_{\pm n\in \bbN}\psi^\pm_n z^{-n},
$
and the central element $\gamma^{\pm 1/2}$, satisfying the defining relations
\begin{align}
&\psi^\pm(z) \psi^\pm(w)= \psi^\pm(w) \psi^\pm(z),\qquad 
\psi^+(z)\psi^-(w)=
\dfrac{g(\gamma^{+1} w/z)}{g(\gamma^{-1}w/z)}\psi^-(w)\psi^+(z),\\
&\psi^+(z)x^\pm(w)=g(\gamma^{\mp 1/2}w/z)^{\mp1} x^\pm(w)\psi^+(z),\\
&\psi^-(z)x^\pm(w)=g(\gamma^{\mp 1/2}z/w)^{\pm1} x^\pm(w)\psi^-(z),\\
&
[x^+(z),x^-(w)]=\dfrac{(1-q)(1-1/t)}{1-q/t}
\bigg( \delta(\gamma^{-1}z/w)\psi^+(\gamma^{1/2}w)-
\delta(\gamma z/w)\psi^-(\gamma^{-1/2}w) \bigg),\\
&G^{\mp}(z/w)x^\pm(z)x^\pm(w)=G^{\pm}(z/w)x^\pm(w)x^\pm(z),
\end{align}
where $\delta(z)=\sum_{n\in \bbZ} z^n$.
\end{dfn}

\begin{prp}\label{prop:Hopf-alg}
The algebra $\calU$ has a Hopf algebra structure defined by the 
coproduct $\Delta$:
\begin{align*}
&\Delta(\gamma^{\pm 1/2})=\gamma^{\pm 1/2} \otimes \gamma^{\pm 1/2},\\
&\Delta (x^+(z))=
x^+(z)\otimes 1+
\psi^-(\gamma_{(1)}^{1/2}z)\otimes x^+(\gamma_{(1)}z),\\
&\Delta (x^-(z))=
x^-(\gamma_{(2)}z)\otimes \psi^+(\gamma_{(2)}^{1/2}z)+1 \otimes x^-(z),\\
&\Delta (\psi^\pm(z))=
\psi^\pm (\gamma_{(2)}^{\pm 1/2}z)\otimes \psi^\pm (\gamma_{(1)}^{\mp 1/2}z),
\end{align*}
where $\gamma_{(1)}^{\pm 1/2}=\gamma^{\pm 1/2}\otimes 1$
and $\gamma_{(2)}^{\pm 1/2}=1\otimes \gamma^{\pm 1/2}$.
We omit the counit $\varepsilon$ and the antipode $a$ since we do not need them here.
\end{prp}

\begin{rmk}
The $\psi^\pm_0$ are central elements in $\calU$.
\end{rmk}

\begin{dfn}
Let $M$ be a left $\calU$-module over $\FF$. If we have 
\begin{align}
&\gamma^{1/2} \alpha= (t/q)^{l_1/4} \alpha,\qquad 
 (\psi^+_0)^{-1}\psi_0^- \alpha= (t/q)^{l_2} \alpha
\end{align}
for any $\alpha\in M$ and for some fixed $l_1,l_2\in \bbZ$,
we call $M$ of level $(l_1,l_2)$.

\end{dfn}

\subsection{Macdonald symmetric functions and Fock space $\calF$}

We basically follow \cite{Ma:1995} for the notations. 
A partition $\lambda$ is a series of 
nonnegative integers $\lambda=(\lambda_1,\lambda_2,\ldots)$ 
such that $\lambda_1\ge\lambda_2\ge\cdots$ with finitely many nonzero entries.
We use the following symbols:
$|\lambda|  \seteq \sum_{i\geq 1} \lambda_i$, 
$n(\lambda) \seteq \sum_{i\geq 1}(i-1)\lambda_i$. If $\lambda_l>0$ and $\lambda_{l+1}=0$,
we write $\ell(\lambda) \seteq l$ and call it  the length of $\lambda$.
The conjugate partition of $\lambda$ is denoted by $\lambda'$ which corresponds to 
the transpose of the diagram $\lambda$.
The empty sequence is denoted by $\emptyset$.
The dominance ordering is defined by $\lambda\ge\mu$ $\Leftrightarrow$
$|\lambda|=|\mu|$ and 
$\sum_{k=1}^i \lambda_k \ge \sum_{k=1}^i \mu_k$ for all $i=1,2,\ldots$.

We also follow \cite{Ma:1995} for the convention of the Young diagram.
Namely, the first coordinate $i$ (the row index) increases as one goes downwards,
and the second coordinate $j$ (the column index) increases 
as one goes rightwards. 
We denote by $\square=(i,j)$ the box located at the coordinate $(i,j)$.
For a box $\square=(i,j)$ and a partition $\lambda$, we use the following notations:
\begin{align*}
i(\square)\seteq i,\quad 
j(\square)\seteq j,\quad
a_{\lambda}(\square)\seteq \lambda_i-j,\quad
\ell_{\lambda}(\square)\seteq \lambda'_j-i. 
\end{align*}

Let $\Lambda$ be the ring of symmetric functions in $x=(x_1,x_2,\ldots)$ over $\bbZ$,
and let $\Lambda_{\bbQ(q,t)}\seteq\Lambda\otimes_{\bbZ}\bbQ(q,t)$. 
Let $m_\lambda$ be the monomial symmetric functions.
Denote the power sum function by $p_n=\sum_{i\geq 1}x_i^n$. For a partition $\lambda$, 
we write
$p_\lambda=\prod_i p_{\lambda_i}$.
Macdonald's scalar product on $\Lambda_{\bbF}$ is
\begin{align}
\langle p_\lambda,p_\mu \rangle_{q,t}=\delta_{\lambda,\mu}
z_\lambda \prod_{i=1}^{\ell(\lambda)} {1-q^{\lambda_i}\over 1-t^{\lambda_i}},\qquad 
z_\lambda=\prod_{i\geq 1} i^{m_i} \cdot m_i!, \label{scalar}
\end{align}
Here we denote by $m_i$ 
the number of entries  in $\lambda$ equal to $i$.
\begin{fct}
The Macdonald  symmetric function $P_\lambda(x;q,t)$ is uniquely characterized by 
the conditions \cite[Chap. VI, (4.7)]{Ma:1995}.
\begin{align*}
& P_\lambda= m_\lambda+\sum_{\mu<\lambda} u_{\lambda\mu}m_\mu
\qquad (u_{\lambda\mu}\in \bbQ(q,t)),\\
& \langle P_\lambda,P_\mu \rangle_{q,t}=0\qquad (\lambda\neq \mu). 
\end{align*}
\end{fct}

Denote 
$Q_\lambda \seteq P_\lambda / \langle P_\lambda,P_\lambda \rangle_{q,t}$.
Then $(Q_\lambda)$ and $(P_\lambda)$ are dual bases of $\Lambda_{\bbF}$.
We have
 $\langle P_\lambda,P_\lambda\rangle_{q,t}= c'_\lambda / c_\lambda$
 where
\begin{align}
c_\lambda \seteq
 \prod_{\square\in \lambda }
 (1-q^{a_\lambda(\square)}t^{\ell_\lambda(\square)+1}),\quad
c'_\lambda \seteq
 \prod_{\square\in \lambda }
 (1-q^{a_\lambda(\square)+1}t^{\ell_\lambda(\square)}). \label{c-lam}
\end{align}

Let $x=(x_1,x_2,\ldots)$ and $y=(y_1,y_2,\ldots)$ are two infinite sets of 
independent indeterminates.
The skew Macdonald polynomials $P_{\lambda/\mu}$ satisfy 
$P_\lambda(x,y)=\sum_\mu P_\mu(x) P_{\lambda/\mu}(y) $
 \cite[Chap. VI, (7.$9'$)]{Ma:1995}.

Let $\calH$ be the Heisenberg algebra over $\bbF$ 
with generators $\{a_n \mid n\in\bbZ\}$ satisfying 
\begin{align*}
 [a_m,a_n]=m\dfrac{1-q^{|m|}}{1-t^{|m|}}\delta_{m+n,0} \, a_0.
\end{align*}
Let  $|0\rangle$ be the vacuum state
satisfying the annihilation conditions for 
the positive Fourier modes $a_n|0\rangle=0$ ($n\in\bbZ_{>0}$).
For a partition $\lambda=(\lambda_1,\lambda_2,\ldots)$, we 
denote 
$\ket{a_{\lambda} }=a_{-\lambda_1} a_{-\lambda_2}\cdots |0\rangle$ for short.
Denote by $\calF$ the Fock space 
having the basis $(\ket{a_{\lambda} })$.

As graded vector spaces, the space of the symmetric functions $\Lambda_{\bbF}$
and the Fock space $\calF$ are isomorphic, and 
we may identify them:
\begin{align}
\calF \simto \Lambda_{\bbF},\quad 
      \ket{ a_{\lambda}} \mapsto p_\lambda.
\end{align}
We give an $\calH$-module structure on $\Lambda_{\bbF}$ by setting
$a_0v=v$ and
\begin{align*}
a_{-n} v=p_n v,\quad
a_{n}  v=n\dfrac{1-q^n}{1-t^n}\dfrac{\partial v}{\partial p_n},\qquad (n>0,v\in \Lambda_{\bbF}).
\end{align*}

Let $\bra{0}$ be the dual vacuum satisfying $\bra{0} a_n=0$ ($n\in \bbZ_{<0}$), and
 $\bra{a_\lambda}=\bra{0}a_{\lambda_1}a_{\lambda_2}\cdots$.
The dual Fock space $\calF^*$ has the basis $(\bra{a_\lambda})$.
We identify symmetric functions with states in $\calF$ (or $\calF^*$) when it is convenient, and 
write $\ket{P_\lambda}$ (or $\bra{P_\lambda}$) for $P_\lambda$ for example.
With this notation we have 
$\bra{P_\lambda}\calO\ket {P_\mu}=\langle P_\lambda, \calO P_\mu \rangle_{q,t}$
for any  $\calO\in U(\calH)$.

\subsection{Level $(0,1)$ module $\calF_u^{(0,1)}$}
\begin{dfn}
Let $\lambda=(\lambda_1,\lambda_2,\ldots)$ be a partition, and $i\in \bbZ_{\geq0}$. 
Set $A^\pm_{\lambda,i} \in \bbQ(q,t), B^\pm_{\lambda}(z) \in \bbQ(q,t)[[z]]$ by 
\begin{align}
&A^+_{\lambda,i}=(1-t)
\prod_{j=1}^{i-1}
{(1-q^{\lambda_i-\lambda_j}t^{-i+j+1})
(1-q^{\lambda_i-\lambda_j+1}t^{-i+j-1})\over 
(1-q^{\lambda_i-\lambda_j}t^{-i+j})
(1-q^{\lambda_i-\lambda_j+1}t^{-i+j})},\label{A+}\\
&A^-_{\lambda,i}=(1-t^{-1})
{1-q^{\lambda_{i+1}-\lambda_i} \over 1-q^{\lambda_{i+1}-\lambda_i+1} t^{-1}}
\prod_{j=i+1}^{\infty}
{(1-q^{\lambda_j-\lambda_i+1}t^{-j+i-1})
(1-q^{\lambda_{j+1}-\lambda_i}t^{-j+i})\over 
(1-q^{\lambda_{j+1}-\lambda_i+1}t^{-j+i-1})
(1-q^{\lambda_j-\lambda_i}t^{-j+i})} ,\label{A-}\\
&
B^+_\lambda(z)=
{1-q^{\lambda_{1}-1} t z \over 1-q^{\lambda_{1}} z}
\prod_{i=1}^{\infty}
{(1-q^{\lambda_i}t^{-i}z)
(1-q^{\lambda_{i+1}-1}t^{-i+1}z)\over 
(1-q^{\lambda_{i+1}}t^{-i}z)
(1-q^{\lambda_i-1}t^{-i+1}z)},\label{B+}\\
&
B^-_\lambda(z)=
{1-q^{-\lambda_{1}+1} t^{-1} z \over 1-q^{-\lambda_{1}} z}
\prod_{i=1}^{\infty}
{(1-q^{-\lambda_i}t^{i}z)
(1-q^{-\lambda_{i+1}+1}t^{i-1}z)\over 
(1-q^{-\lambda_{i+1}}t^{i}z)
(1-q^{-\lambda_i+1}t^{i-1}z)}.\label{B-}
\end{align}
\end{dfn}

Note that if $\lambda_i=\lambda_{i-1}$ then $A^+_{\lambda,i}=0$, and
if $\lambda_i=\lambda_{i+1}$ then $A^-_{\lambda,i}=0$.
If $\lambda_i<\lambda_{i-1}$, we may obtain a new partition by 
adding one box to the $i$-th row, and we denote it by
$\lambda+{\bf 1}_i=
(\lambda_1,\lambda_2,\ldots,\lambda_{i-1},\lambda_i+1,\lambda_{i+1},\ldots)
$
for simplicity. 
If $\lambda_i>\lambda_{i+1}$, we may obtain a new partition by 
removing one box from the $i$-th row, and we write
$
\lambda-{\bf 1}_i=
(\lambda_1,\lambda_2,\ldots,\lambda_{i-1},\lambda_i-1,\lambda_{i+1},\ldots).
$

\begin{prp} Let $u$ be an indeterminate.
We can endow a left $\calU$-module structure over $\FF$
to $\calF$ by setting
\begin{align}
&
\gamma^{1/2} P_\lambda=P_\lambda,\\
&
x^+(z) P_\lambda=
\sum_{i=1}^{\ell(\lambda)+1} 
A^+_{\lambda,i}\,
\delta(q^{\lambda_i}t^{-i+1}u/z) 
P_{\lambda+{\bf 1}_i},\\
&
x^-(z) P_\lambda=q^{1/2}t^{-1/2}
\sum_{i=1}^{\ell(\lambda)} 
A^-_{\lambda,i}\,
\delta(q^{\lambda_i-1}t^{-i+1}u/z)
P_{\lambda-{\bf 1}_i},\\
&
\psi^+(z)P_\lambda=q^{1/2}t^{-1/2}
B^+_\lambda(u/z)P_\lambda,\\
&
\psi^-(z)P_\lambda=q^{-1/2}t^{1/2}
\,B^-_\lambda(z/u)P_\lambda.
\end{align}
This is a level $(0,1)$ module. 
We denote this $\calU$-module by $\calF^{(0,1)}_u$.
\end{prp}
 This was obtained in \cite{FT:2009}, \cite{FFJMM1:2010}.

\subsection{Level $(1,N)$ module $\calF_u^{(1,N)}$}
\begin{dfn}
Set 
\begin{align}
&\eta(z)=
\exp\Big( \sum_{n=1}^{\infty} \dfrac{1-t^{-n}}{n}a_{-n} z^{n} \Big)
\exp\Big(-\sum_{n=1}^{\infty} \dfrac{1-t^{n} }{n}a_n    z^{-n}\Big),\\
&\xi(z)=
\exp\Big(-\sum_{n=1}^{\infty} \dfrac{1-t^{-n}}{n}
q^{-n/2}t^{n/2}
a_{-n} z^{n}\Big)
\exp\Big( \sum_{n=1}^{\infty} \dfrac{1-t^{n}}{n} q^{-n/2}t^{n/2} a_n z^{-n}\Big),\\
&\varphi^{+}(z)=
\exp\Big(
 -\sum_{n=1}^{\infty} \dfrac{1-t^{n}}{n} (1-t^n q^{-n})
 q^{n/4}t^{-n/4} a_n z^{-n}
    \Big),
\\
&\varphi^{-}(z)=
\exp\Big(
 \sum_{n=1}^{\infty} \dfrac{1-t^{-n}}{n} (1-t^n q^{-n}) q^{n/4}t^{-n/4} a_{-n}z^{n}
    \Big).
\end{align}
\end{dfn}

\begin{prp}
Let $u$ be an indeterminate, and $N\in \bbZ$.
We can endow a left $\calU$-module structure over $\FF$ to $\calF$ by setting 
\begin{align}
&
\gamma^{1/2} P_\lambda=(t/q)^{1/4} P_\lambda,\\
&
x^+(z) P_\lambda= 
u z^{-N} q^{-N/2}t^{N/2}\eta(z) P_\lambda,\\
&
x^-(z) P_\lambda=
u^{-1} z^{N}q^{N/2}t^{-N/2} \xi(z) P_\lambda,\\
&
\psi^+(z)P_\lambda=
q^{N/2}t^{-N/2}
\varphi^+(z) P_\lambda,\\
&
\psi^-(z)P_\lambda=
q^{-N/2}t^{N/2}\varphi^-(z) P_\lambda.
\end{align}
This is a level $(1,N)$ module.
We denote this $\calU$-module by $\calF^{(1,N)}_u$.
\end{prp}

This is obtained  as an easy modification of the representation constructed in \cite{FHHSY:2009}.


\section{Trivalent Intertwining Operators $\Phi$ and $\Phi^*$}

\subsection{Intertwining operator $\Phi$}
Let $N\in \bbZ$ and $u,v,w$ be independent indeterminates.
\begin{dfn}\label{def-phi}
Let $\Phi=\Phi\left[(1,N+1),w \atop (0,1),v\, ;\, (1,N),u \right]$ 
be the trivalent intertwining operator satisfying the conditions
\begin{align}
&
\Phi:\calF^{(0,1)}_v\otimes \calF^{(1,N)}_u\longrightarrow
\calF^{(1,N+1)}_w,\\
& a \Phi= \Phi \Delta(a) \qquad (\forall a\in \calU). \label{IR-1}
\end{align}
Introduce the components $\Phi_\lambda$ by setting 
\begin{align}
\Phi_\lambda (\alpha)=\Phi(P_\lambda \otimes \alpha )
\qquad (\forall 
P_\lambda \otimes \alpha \in \calF^{(0,1)}_v\otimes \calF^{(1,N)}_u).
\end{align}
We normalize $\Phi$ by requiring 
$\Phi_\emptyset (1)=1+\cdots $.
\end{dfn} 

\begin{lem}\label{lemma-1}
The intertwining relations (\ref{IR-1}) read
\begin{align}
&
\sum_{i=1}^{\ell(\lambda)+1}
A^+_{\lambda,i}\, \delta(q^{\lambda_i} t^{-i+1}v/z) \Phi_{\lambda+{\bf 1}_i}+
q^{-1/2}t^{1/2} B^-_\lambda(z/v) \Phi_\lambda x^+(z)=x^+(z)  \Phi_\lambda,\\
&
q^{1/2}t^{-1/2}
\sum_{i=1}^{\ell(\lambda)}
A^-_{\lambda,i}\, \delta(q^{\lambda_i-1} t^{-i+1}v/z) \Phi_{\lambda-{\bf 1}_i}
\psi^+(q^{1/4}t^{-1/4}z)+
\Phi_\lambda x^-(q^{1/2}t^{-1/2}z)=x^-(q^{1/2}t^{-1/2}z)\Phi_\lambda,\\
&
q^{1/2}t^{-1/2}
B^+_\lambda(v/z)\Phi_\lambda \psi^+(q^{1/4}t^{-1/4}z )
=\psi^+(q^{1/4}t^{-1/4}z )
\Phi_\lambda ,\\
&
q^{-1/2}t^{1/2} B^-_\lambda(z/v)\Phi_\lambda \psi^-(q^{-1/4}t^{1/4}z )
=\psi^-(q^{-1/4}t^{1/4}z )
\Phi_\lambda .
\end{align}
\end{lem}

\begin{thm}\label{thm-1}
The normalized intertwining operator $\Phi$ exists only when $w=-v u$. 
In this case, it is determined uniquely
and written 
in terms of the Heisenberg generators as
\begin{align}
&\Phi_\lambda\left[(1,N+1),-vu \atop (0,1),v\, ;\, (1,N),u \right]
=t(\lambda,u,v,N) \tPhi_\lambda(v),\\
&t(\lambda,u,v,N)=(-vu)^{|\lambda|} (-v)^{-(N+1)|\lambda|}
f_\lambda^{-N-1}.
\label{t(lam)}
\end{align}
Here we have used the notations
\begin{align}
&
\tPhi_\lambda(v)={q^{n(\lambda')}\over c_\lambda}
:\Phi_{\emptyset}(v) \eta_\lambda(v):,\\
&\tPhi_{\emptyset}(v) =
 \exp \Bigl(
 -\sum_{n=1}^{\infty} \dfrac{1}{n}\dfrac{1}{1-q^n} a_{-n}v^n 
      \Bigr)
 \exp\Bigl(
-  \sum_{n=1}^{\infty} \dfrac{1}{n}\dfrac{q^n}{1-q^{n}} a_{n}v^{-n}
     \Bigr),\\
&
\eta_\lambda(v)=\,
:\prod_{i=1}^{\ell(\lambda)}\prod_{j=1}^{\lambda_i}
\eta(q^{j-1}t^{-i+1} v):,
\end{align}
where the symbol $: \cdots:$ denotes the usual normal ordering, 
and $f_\lambda$ is Taki's framing factor (\ref{f-lam}).

\end{thm}

A proof of this will be given in Section \ref{section-6}.

\subsection{Intertwining operator $\Phi^*$}

\begin{dfn}\label{def-phis}
Let $\Phi^*=\Phi^*\left[ (1,N),v\, ;\, (0,1),u  \atop (1,N+1),w \right]$ 
be the trivalent intertwining operator satisfying the conditions
\begin{align}
&
\Phi^*:
\calF^{(1,N+1)}_w\longrightarrow
\calF^{(1,N)}_v\otimes 
\calF^{(0,1)}_u,\\
& \Delta(a) \Phi^*= \Phi^* a \qquad (\forall a\in \calU).\label{IR-2}
\end{align}
Introduce the components $\Phi_\lambda$ by setting 
\begin{align}
\Phi^*( \alpha )=\sum_\lambda \Phi^*_\lambda (\alpha)\otimes Q_\lambda
\qquad (\forall \alpha \in \calF^{(1,N+1)}_w).
\end{align}
We normalize $\Phi^*$ by requiring 
$\Phi^*_\emptyset (1)=1+\cdots $.
\end{dfn} 

\begin{lem}\label{lemma-2}
The intertwining relations (\ref{IR-2}) read
\begin{align}
&
\Phi^*_\lambda x^+(q^{1/2}t^{-1/2}z)=x^+(q^{1/2}t^{-1/2}z)\Phi^*_\lambda
-\psi^-(q^{1/4}t^{-1/4}z)
\sum_{i=1}^{\ell(\lambda)}
q A^-_{\lambda,i}\, \delta(q^{\lambda_i-1} t^{-i+1}u/z) 
\Phi^*_{\lambda-{\bf 1}_i},\\
&
\Phi^*_\lambda x^-(z)
=q^{1/2}t^{-1/2}B^+_\lambda(u/z)x^-(z)\Phi^*_\lambda-
q^{1/2}t^{-1/2}
\sum_{i=1}^{\ell(\lambda)+1}
q^{-1}A^+_{\lambda,i}\, \delta(q^{\lambda_i} t^{-i+1}u/z) \Phi^*_{\lambda+{\bf 1}_i},\\
&
\Phi^*_\lambda  \psi^+(q^{-1/4}t^{1/4}z )
=
q^{1/2}t^{-1/2}
B^+_\lambda(u/z)\psi^+(q^{-1/4}t^{1/4}z )\Phi^*_\lambda ,\\
&
\Phi^*_\lambda  \psi^-(q^{1/4}t^{-1/4}z )=
q^{-1/2}t^{1/2} B^-_\lambda(z/u) \psi^-(q^{1/4}t^{-1/4}z )\Phi^*_\lambda .
\end{align}
\end{lem}

\begin{thm}\label{thm-2}
The intertwining operator $\Phi^*$ exists uniquely only when $w=-v u$. 
In this case, it is written 
in terms of the Heisenberg generators as
\begin{align}
&\Phi^*_\lambda\left[ (1,N),v\, ;\, (0,1),u  \atop (1,N+1),-vu \right]=t^*(\lambda,u,v,N) \tPhi^*_\lambda(u),\\
&t^*(\lambda,u,v,N)=
(q^{-1} v)^{-|\lambda|} (-u)^{N|\lambda|}
f_\lambda^{N}.
\label{ts(lam)}
\end{align}
Here $f_\lambda$ is given in (\ref{f-lam}), and 
\begin{align}
&
 \tPhi^*_\lambda(u)=
 {q^{n(\lambda')}\over c_\lambda}
:\tPhis_{\emptyset}(u) \xi_\lambda(u):,\\
     &
\tPhis_{\emptyset}(u) =
 \exp \Bigl(
\sum_{n=1}^{\infty} \dfrac{1}{n}\dfrac{1}{1-q^n} q^{-n/2}t^{n/2}a_{-n}u^n 
      \Bigr)
 \exp\Bigl(
  \sum_{n=1}^{\infty} \dfrac{1}{n}\dfrac{q^n}{1-q^{n}} q^{-n/2}t^{n/2}a_{n}u^{-n}
     \Bigr),\\
     &\xi_\lambda(u)=\,
:\prod_{i=1}^{\ell(\lambda)}\prod_{j=1}^{\lambda_i}
\xi(q^{j-1}t^{-i+1} u):.
\end{align} 
\end{thm}

In Section \ref{section-6}, we give a proof of this.

\section{Identification with refined topological vertex}

\subsection{Notations}
Let $s_\lambda(x)\in \Lambda_\bbZ$ be the Schur function, and $c_{\lambda\mu}^\nu$ 
be the Littlewood-Richardson coefficient determined by 
$s_\lambda s_\mu =\sum_\nu c_{\lambda\mu}^\nu s_\nu$.
The skew Schur function is defined by $s_{\lambda/\mu}=\sum_\nu c^\lambda_{\mu\nu} s_\nu$,
and we have  $s_\lambda(x,y)=\sum_\mu s_{\lambda/\mu}(x) s_\mu(y)$
 \cite[Chap. I. (5.9)]{Ma:1995}.
 Let $S_\lambda(x;q,t)\in \Lambda_\bbF$ be the dual of $s_\lambda$ with respect to the 
scalar product (\ref{scalar}), namely  $\langle S_\lambda(q,t),s_\mu\rangle_{q,t}=\delta_{\lambda,\mu} $.
Set $S_{\lambda/\mu}(x;q,t)=\sum_\nu c^\lambda_{\mu\nu} S_\nu(x;q,t)$.
We have $S_\lambda(x,y;q,t)=\sum_\mu S_{\lambda/\mu}(x;q,t) S_\mu(y;q,t)$.

Recall the $\bbF$-algebra endomorphism $\omega_{u,v}$ of Macdonald  \cite[Chap. VI. (2.14)]{Ma:1995}.
\begin{align}
\omega_{u,v} (p_n)=-(-1)^n {1-u^n\over 1-v^n} p_n.
\end{align}
It is convenient to have  two operations $\iota$ and $\varepsilon^\pm_\lambda$ acting on $\Lambda_\bbF$ 
introduced in
\cite{AK:2009}.
The $\iota$ is defined to be the involution  on $\Lambda_\bbF$ given by
\begin{align}
\iota: \Lambda_\bbF\rightarrow \Lambda_\bbF,
\qquad \iota(p_n)=-p_n\quad (n\in \bbZ_{>0}).
\end{align}
The $\varepsilon^\pm_\lambda=\varepsilon^\pm_{\lambda,q,t}$
is defined to be the algebra homomorphism 
\begin{align}
\varepsilon^\pm_\lambda:\Lambda_\bbF\rightarrow\bbF,
\qquad 
\varepsilon^\pm_\lambda (p_n)=
\sum_{i=1}^\infty (q^{\pm \lambda_i n}-1)t^{\mp (i-1/2)n}+{t^{\mp n/2}\over 1-t^{\mp n}}, \label{epsilon-pm}
\end{align}
For any symmetric functions,
we shall use the shorthand notations such as
\begin{align}
&
\varepsilon^\pm_\lambda (s_\mu)=s_\mu(q^{\pm \lambda}t^{\pm \rho}),
\qquad
\varepsilon^\pm_\lambda (\iota s_\mu)=\iota s_\mu(q^{\pm \lambda}t^{\pm \rho}),
\end{align}
since we may have the interpretation $\rho=(-1/2,-3/2,-5/2,\ldots)$ in mind.

We have
\begin{align}
&\iota s_\lambda(x)=s_{\lambda'}(-x)=(-1)^{|\lambda|} s_{\lambda'}(x),\\
&
S_\lambda(x;q,t)=
\iota \,\omega_{t,q} s_\lambda(-x),
\end{align}
and
\begin{align}
&
\varepsilon^+_{\lambda,q,t} (p_n(x))=
\varepsilon^-_{\lambda',t,q}\, \omega_{q,t}(p_n(-q^{-1/2}t^{1/2}x) ).
\end{align}
Hence
\begin{align}
&
\varepsilon^+_{\lambda,q,t} 
\iota S_\mu(x;q,t)= (q^{1/2}t^{-1/2})^{-|\mu|} 
\varepsilon^-_{\lambda',t,q}\, 
s_\mu(x). \label{S(q,t)}
\end{align}
In the shorthand notation this is written as 
$\iota S_\mu(q^\lambda t^\rho ;q,t)= (q^{1/2}t^{-1/2})^{-|\mu|} 
s_\mu(t^{-\lambda'} q^{-\rho})
$.

\subsection{Matrix elements of $\Phi$ and $\Phi^*$}
We have simple but important formulas 
which essentially control the property of our intertwining operators.
\begin{prp}\label{evalu}
We have
\begin{align}
&
:\Phi_\emptyset (q^{1/2}v) \eta_\lambda(q^{1/2}v):\label{Phi-1}\\
&=
\exp\left(\sum_{n=1}^\infty {1\over n}{1-t^n\over 1-q^n} a_{-n} (q^{1/2 }t^{-1/2})^n v^n  
\varepsilon^+_\lambda(p_n) \right)
\exp\left(-\sum_{n=1}^\infty {1\over n}{1-t^n\over 1-q^n} a_{n} (q^{1/2 }t^{-1/2})^n v^{-n}  
\varepsilon^-_\lambda(p_n) \right),\nonumber\\
&
:\Phi^*_\emptyset (q^{1/2}u) \xi_\lambda(q^{1/2}u):\label{Phis-1}\\
&=
\exp\left(-\sum_{n=1}^\infty {1\over n}{1-t^n\over 1-q^n} a_{-n}u^n  
\varepsilon^+_\lambda(p_n) \right)
\exp\left(\sum_{n=1}^\infty {1\over n}{1-t^n\over 1-q^n} a_{n}  u^{-n}  
\varepsilon^-_\lambda(p_n) \right).\nonumber
\end{align}
\end{prp}

\begin{cor}\label{P-Phi-Q}
We have
\begin{align}
&\bra{S_\nu(q,t)} :\tPhi_\emptyset (q^{1/2}v) \eta_\lambda(q^{1/2}v):\ket{s_\mu}\label{sk-1-s}\\
&=
v^{|\nu|-|\mu|}(q^{1/2}t^{-1/2})^{|\nu|+|\mu|}
\sum_\sigma S_{\nu/\sigma}(q^\lambda t^\rho ;q,t)
 \iota s_{\mu/\sigma}(q^{-\lambda}t^{-\rho}) 
(q^{1/2}t^{-1/2})^{-2|\sigma|},\nonumber\\
&\bra{S_\nu(q,t)} :\tPhis_\emptyset (q^{1/2}u) \xi_\lambda(q^{1/2}u):\ket{s_\mu}\label{sk-2-s}\\
&=
u^{|\nu|-|\mu|}
\sum_\sigma \iota S_{\nu/\sigma}(q^\lambda t^\rho ;q,t)
 s_{\mu/\sigma}(q^{-\lambda}t^{-\rho}) 
,\nonumber
 \end{align}
 and
 \begin{align}
&\bra{P_\nu} :\tPhi_\emptyset (q^{1/2}v) \eta_\lambda(q^{1/2}v):\ket{P_\mu}\label{sk-1}\\
&=
v^{|\nu|-|\mu|}(q^{1/2}t^{-1/2})^{|\nu|+|\mu|}
\sum_\sigma P_{\nu/\sigma}(q^\lambda t^\rho )
 \iota P_{\mu/\sigma}(q^{-\lambda}t^{-\rho}) 
 \langle P_\sigma,P_\sigma\rangle_{q,t}(q^{1/2}t^{-1/2})^{-2|\sigma|},\nonumber\\
&\bra{P_\nu} :\tPhis_\emptyset (q^{1/2}u) \xi_\lambda(q^{1/2}u):\ket{P_\mu}\label{sk-2}\\
&=
u^{|\nu|-|\mu|}
\sum_\sigma \iota P_{\nu/\sigma}(q^\lambda t^\rho )
 P_{\mu/\sigma}(q^{-\lambda}t^{-\rho}) 
 \langle P_\sigma,P_\sigma\rangle_{q,t}.\nonumber
\end{align}
\end{cor}

\begin{proof}
{}From Proposition \ref{evalu}
we have
\begin{align*}
&\bra{0} \prod_{i} a_{\nu_i}
 :\tPhi_\emptyset (q^{1/2}v) \eta_\lambda(q^{1/2}v):
 \prod_{j} a_{-\mu_j}\ket{0}\\
&\qquad =
\bra{0}\prod_{i}
(a_{\nu_i}+(q^{1/2}t^{-1/2}v)^{\nu_i}\varepsilon^+_\lambda(p_{\nu_i}))
\cdot 
\prod_{j}
(a_{-\mu_j}- (q^{1/2}t^{-1/2}v^{-1})^{\mu_j}\varepsilon^-_\lambda(p_{\mu_j}))
 \ket{0},\nonumber \\
&
\bra{0} \prod_{i} a_{\nu_i}
 :\tPhis_\emptyset (q^{1/2}u) \xi_\lambda(q^{1/2}u): 
  \prod_{j} a_{-\mu_j}\ket{0}
\\
&\qquad =
\bra{0} \prod_{i} 
(a_{\nu_i}-u^{\nu_i} \varepsilon^+_\lambda(p_{\nu_i}))\cdot
 \prod_{j} 
(a_{-\mu_j}+u^{-\mu_j}\varepsilon^-_\lambda(p_{\mu_j})) \ket{0}.
\end{align*}
Then (\ref{sk-1-s}), (\ref{sk-2-s}), (\ref{sk-1}) and (\ref{sk-2}) follow from 
the property of the skew functions.
\end{proof}

\subsection{Topological vertex of Iqbal-Kozcaz-Vafa}

\begin{dfn}[Iqbal-Kozcaz-Vafa]
The refined topological vertex is defined by
\begin{align}
C^{{\rm (IKV)}}_{\lambda\mu\nu}(t,q)=
\left(q\over t\right)^{||\mu||^2\over 2} t^{\kappa(\mu)\over 2} q^{||\nu||^2\over 2}
{1\over c_\lambda}
\sum_\eta 
\left(q\over t\right)^{|\eta|+|\lambda|-|\mu|\over 2} 
s_{\lambda'/\eta}(t^{-\rho}q^{-\nu})s_{\mu/\eta}(t^{-\nu'}q^{-\rho}),
\end{align}
where $c_\lambda$ is defined in (\ref{c-lam}), $||\lambda||^2=\sum_i\lambda_i^2$,
$\kappa(\lambda)=\sum_i \lambda_i(\lambda_i+1-2 i)$.
\end{dfn}

\begin{prp}\label{mat-el-IKV}
The matrix elements of the intertwining operators $\Phi$ and $\Phi^*$
are written in terms of the refined topological vertex as
\begin{align}
&{1\over \langle P_\lambda,P_\lambda\rangle_{q,t}}
\bra{S_\mu(q,t)}\Phi_\lambda\left[  (1,N+1),-vu \atop (0,1),v; (1,N) ,u\right] 
\ket{s_\nu}\\
&\qquad
=\left( q^{-1/2}u\over (-v)^{N} \right)^{|\lambda|}
f_\lambda^{-N}\cdot \,\,
(-q^{-1/2}v)^{-|\nu|} f_\nu \cdot 
(t^{-1/2}v)^{|\mu|} \cdot 
 (-1)^{|\mu|+|\nu|+|\lambda|}
\,\,C^{\rm(IKV)}_{\mu\nu' \lambda'}(q,t),\nonumber\\
&
\bra{S_\nu(q,t)}\Phi^*_\lambda\left[  (1,N) ,v;(0,1),u \atop (1,N+1),-vu\right] 
\ket{s_\mu}\\
&\qquad
=\left( (-u)^N\over q^{-1/2} v \right)^{|\lambda|}
f_\lambda^N \cdot \,\,
(-q^{-1/2}u)^{|\nu|} f_\nu^{-1} \cdot 
(t^{-1/2}u)^{-|\mu|} \cdot 
 \,C^{\rm (IKV)}_{\mu'\nu \lambda}(t,q).\nonumber
\end{align}
\end{prp}

\begin{proof}
Using Corollary \ref{P-Phi-Q} and (\ref{S(q,t)}) we have the results.
\end{proof}

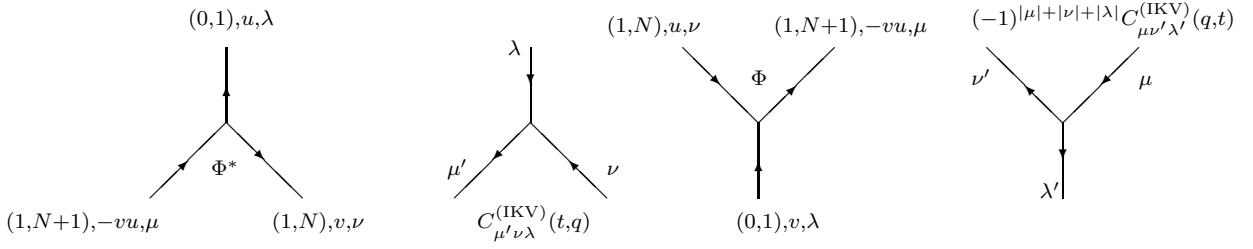
\begin{figure}
\begin{center}
\begin{picture}(60,90)(90,-40)\setlength{\unitlength}{1.mm}

\put(-10,00){\vector(0,1){5}}\put(-10,5){\line(0,1){5}}
\put(-20,-10){\vector(1,1){5}}\put(-15,-5){\line(1,1){5}}
\put(-10,0){\vector(1,-1){5}}\put(-5,-5){\line(1,-1){5}}

\put(30,10){\vector(0,-1){5}}\put(30,5){\line(0,-1){5}}
\put(20,-10){\line(1,1){5}}\put(30,0){\vector(-1,-1){5}}
\put(30,0){\line(1,-1){5}}\put(40,-10){\vector(-1,1){5}}

\put(50,10){\vector(1,-1){5}}\put(55,5){\line(1,-1){5}}
\put(60,-10){\vector(0,1){5}}\put(60,-5){\line(0,1){5}}
\put(60,0){\vector(1,1){5}}\put(65,5){\line(1,1){5}}

\put(90,10){\line(1,-1){5}}\put(100,0){\vector(-1,1){5}}
\put(100,0){\vector(0,-1){5}}\put(100,-10){\line(0,1){5}}
\put(110,10){\vector(-1,-1){5}}\put(105,5){\line(-1,-1){5}}






\put(-12,-7){$\scriptstyle \Phi^*$}

\put(59,5){$\scriptstyle \Phi$}

\put(23,-14){$\scriptstyle C^{\rm(IKV)}_{\mu'\nu\lambda}(t,q)$}

\put(88,13){$\scriptstyle (-1)^{|\mu|+|\nu|+|\lambda|} C^{\rm(IKV)}_{\mu\nu'\lambda'}(q,t) $ }

\put(-39,-14){$\scriptstyle (1,N+1),-vu, \mu$}
\put(-4,-14){$\scriptstyle (1,N),v,\nu$}
\put(-15,13){$\scriptstyle (0,1),u,\lambda$}

\put(19,-7){$\scriptstyle \mu'$}
\put(40,-7){$\scriptstyle \nu$}
\put(27,9){$\scriptstyle \lambda$}

\put(40,12){$\scriptstyle (1,N),u,\nu$}
\put(57,-14){$\scriptstyle (0,1),v,\lambda$}
\put(62,12){$\scriptstyle (1,N+1),-vu,\mu$}

\put(88,5){$\scriptstyle \nu'$}
\put(97,-10){$\scriptstyle \lambda'$}
\put(110,5){$\scriptstyle \mu$}

\end{picture}
\caption{Comparison between  $\Phi,\Phi^*$ and 
 $(-1)^{|\mu|+|\nu|+|\lambda|}C^{\rm(IKV)}_{\mu\nu'\lambda'}(q,t) ,
C^{\rm(IKV)}_{\mu'\nu\lambda}(t,q)$.}
\label{comparison-IKV}
\end{center}
\end{figure}

\begin{rmk}
Note that in the formulation of Iqbal-Kozcaz-Vafa, 
the transpose of the partition is assigned to each outgoing edge. 
To identify the refined topological vertices with vertices for $\Phi$, $\Phi^*$,
all the arrows should be reversed as shown in Fig. \ref{comparison-IKV}.
\end{rmk}

\subsection{Topological vertex of Awata-Kanno}

\begin{dfn}[Awata-Kanno]
The refined topological vertices ${C_{\mu\lambda}}^\nu(q,t), {C^{\mu\lambda}}_\nu(q,t)$
are defined by
\begin{align}
&
{C_{\mu\lambda}}^\nu(q,t)=
P_\lambda(t^\rho;q,t)
\sum_\sigma \iota P_{\mu'/\sigma'}(-t^{\lambda'}q^{\rho};t,q)
P_{\nu/\sigma}(q^\lambda t^\rho;q,t) (q^{1/2}/t^{1/2})^{|\sigma|-|\nu|}
f_\nu(q,t)^{-1},\\
 &
{C^{\mu\lambda}}_\nu(q,t)=
(-1)^{|\lambda|+|\mu|+|\nu|}
{C_{\mu'\lambda'}}^{\nu'}(t,q)
\end{align}
where $f_\lambda$ being defined in (\ref{f-lam}).
\end{dfn}

\begin{prp}\label{mat-el}
The matrix elements of the intertwining operators $\Phi$ and $\Phi^*$
are written in terms of the refined topological vertices as
\begin{align}
&{1\over \langle P_\lambda,P_\lambda\rangle_{q,t}}
\bra{\iota P_\mu}\Phi_\lambda\left[  (1,N+1),-vu \atop (0,1),v; (1,N) ,u\right] 
\ket{\iota Q_\nu}\\
&\qquad
=\left( -t^{1/2}u\over q(-v)^{N} \right)^{|\lambda|}
f_\lambda^{-N}\,\,
(t^{-1/2}v)^{|\mu|-|\nu|} 
f_\nu^{-1} \,\,{C^{\mu\lambda}}_\nu(q,t),\nonumber\\
&
\bra{\iota P_\nu}\Phi^*_\lambda\left[  (1,N) ,v;(0,1),u \atop (1,N+1),-vu\right] 
\ket{\iota Q_\mu}\\
&\qquad
=\left( q(-u)^N\over -t^{1/2} v \right)^{|\lambda|}
f_\lambda^N\,\,
(t^{-1/2}u)^{-|\mu|+|\nu|} 
f_\nu \,\,{C_{\mu\lambda}}^\nu(q,t).\nonumber
\end{align}
\end{prp}

\begin{proof}
Note that we have 
\begin{align}
&
{C_{\mu\lambda}}^\nu(q,t)=
{(-1)^{|\lambda|}q^{n(\lambda')}t^{|\lambda|/2}\over c_\lambda}
(q^{1/2}t^{-1/2})^{|\mu|-|\nu|}f_\nu^{-1} {1\over  \langle P_\mu,P_\mu\rangle_{q,t}}\\
&\qquad \qquad\times
\sum_\sigma P_{\nu/\sigma}(q^\lambda t^\rho )
\iota  P_{\mu/\sigma}(q^{-\lambda}t^{-\rho}) 
 \langle P_\sigma,P_\sigma\rangle_{q,t},\nonumber\\
 &
{C^{\mu\lambda}}_\nu(q,t)=
{q^{|\lambda|/2}t^{n(\lambda)}\over c'_\lambda}
(q^{1/2}t^{-1/2})^{2|\nu|}f_\nu {1\over  \langle P_\nu,P_\nu\rangle_{q,t}}\\
&\qquad\qquad\times
\sum_\sigma
\iota  P_{\mu/\sigma}(q^\lambda t^\rho )
 P_{\nu/\sigma}(q^{-\lambda}t^{-\rho}) 
 \langle P_\sigma,P_\sigma\rangle_{q,t}
 (q^{1/2}t^{-1/2})^{-2|\sigma|} .\nonumber 
 \end{align}
Using Corollary \ref{P-Phi-Q} we have the results.
\end{proof}

\begin{rmk}
We note that all the vertical arrows should be get reversed to establish a correspondence
between $\Phi,\Phi^*$ and  ${C^{\mu\lambda}}_\nu,{C_{\mu\lambda}}^\nu$
as seen in Fig. \ref{comparison}.
\end{rmk}

\begin{figure}
\begin{center}
\begin{picture}(60,90)(90,-40)\setlength{\unitlength}{1.mm}

\put(-10,00){\vector(0,1){5}}\put(-10,5){\line(0,1){5}}
\put(-20,-10){\vector(1,1){5}}\put(-15,-5){\line(1,1){5}}
\put(-10,0){\vector(1,-1){5}}\put(-5,-5){\line(1,-1){5}}

\put(30,10){\vector(0,-1){5}}\put(30,5){\line(0,-1){5}}
\put(20,-10){\vector(1,1){5}}\put(25,-5){\line(1,1){5}}
\put(30,0){\vector(1,-1){5}}\put(35,-5){\line(1,-1){5}}

\put(50,10){\vector(1,-1){5}}\put(55,5){\line(1,-1){5}}
\put(60,-10){\vector(0,1){5}}\put(60,-5){\line(0,1){5}}
\put(60,0){\vector(1,1){5}}\put(65,5){\line(1,1){5}}

\put(90,10){\vector(1,-1){5}}\put(95,5){\line(1,-1){5}}
\put(100,0){\vector(0,-1){5}}\put(100,-5){\line(0,-1){5}}
\put(100,0){\vector(1,1){5}}\put(105,5){\line(1,1){5}}






\put(-12,-7){$\scriptstyle \Phi^*$}

\put(59,5){$\scriptstyle \Phi$}

\put(27,-14){$\scriptstyle {C_{\mu\lambda}}^\nu$}

\put(97,13){$\scriptstyle {C^{\mu\lambda}}_\nu$}

\put(-39,-14){$\scriptstyle (1,N+1),-vu, \mu$}
\put(-4,-14){$\scriptstyle (1,N),v,\nu$}
\put(-15,13){$\scriptstyle (0,1),u,\lambda$}

\put(19,-7){$\scriptstyle \mu$}
\put(40,-7){$\scriptstyle \nu$}
\put(27,9){$\scriptstyle \lambda$}

\put(40,12){$\scriptstyle (1,N),u,\nu$}
\put(57,-14){$\scriptstyle (0,1),v,\lambda$}
\put(62,12){$\scriptstyle (1,N+1),-vu,\mu$}

\put(88,5){$\scriptstyle \nu$}
\put(97,-10){$\scriptstyle \lambda$}
\put(110,5){$\scriptstyle \mu$}

\end{picture}
\caption{Comparison between  $\Phi,\Phi^*$ and  ${C^{\mu\lambda}}_\nu,{C_{\mu\lambda}}^\nu$.}
\label{comparison}
\end{center}
\end{figure}
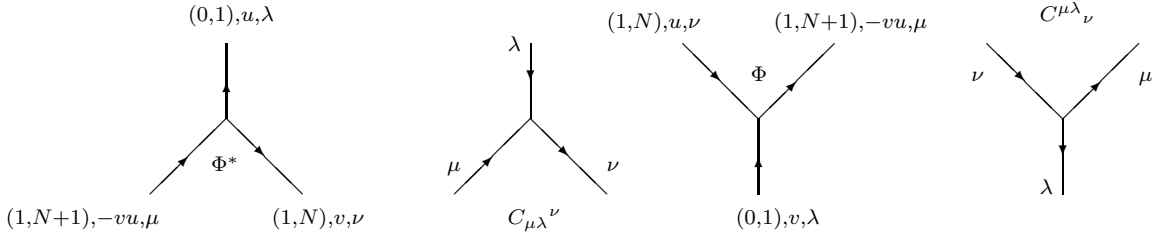

\subsection{Gluing rules}

Consider a trivalent vertex with edges, say,  $i,j$ and $k$,
with two component vectors 
${\bf v}_i,{\bf v}_j,{\bf v}_k\in \bbZ^2$ attached respectively (see Fig. \ref{gluing} (a)).
Here we regard all the vectors being outgoing. 
We assume that they satisfy the (Calabi-Yau and smoothness) conditions 
\begin{align}
{\bf v}_i+{\bf v}_j+{\bf v}_k={\bf 0}, 
\qquad
{\bf v}_i\wedge {\bf v}_j=1,\label{sum}
\end{align}
where we have used the notation $(a,b)\wedge (c,d)=ad-bc$.
Note that these mean ${\bf v}_j\wedge {\bf v}_k={\bf v}_k\wedge {\bf v}_i=1$.

\begin{dfn} \label{gluing-rules}
Let ${\bf v}_i,{\bf v}_j,{\bf v}_k,{\bf v}_{i'},{\bf v}_{j'}\in \bbZ^2$,
and consider a graph as in Fig. \ref{gluing} (b).
Let $\lambda_k$ be a partition, and $Q_k$ be a parameter 
(K\"ahler parameter).
To the internal edge, with the data ${\bf v}_k$, $\lambda_k$, $Q_k$ attached, 
we associate the `gluing factor'
\begin{align}
Q_k^{|\lambda_k|} (f_{\lambda_k})^{ {\bf v}_i\wedge  {\bf v}_{i'}}.\label{factor}
\end{align}
The refined topological vertices are contracted 
by multiplying the gluing factor
and making summation with respect to the repeated indices.

\end{dfn}

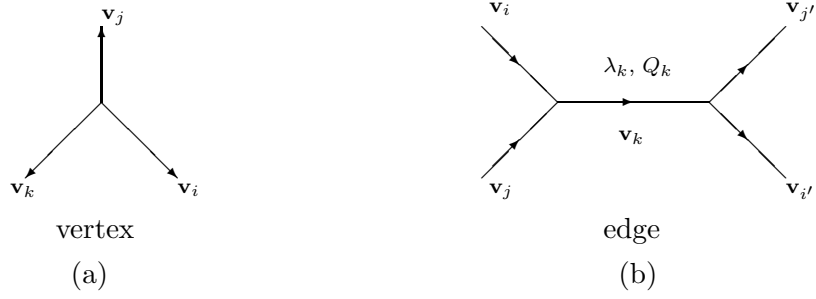
\begin{figure}
\begin{center}
\begin{picture}(60,100)(90,-60)\setlength{\unitlength}{1.mm}

\put(0,0){\vector(0,1){10}}
\put(0,0){\vector(1,-1){10}}
\put(0,0){\vector(-1,-1){10}}

\put(10,-12){$\scriptstyle  {\bf v}_i$}
\put(0,11){$\scriptstyle  {\bf v}_j$}
\put(-12,-12){$\scriptstyle {\bf v}_k$}

\put(50,10){\vector(1,-1){5}}\put(55,5){\line(1,-1){5}}
\put(50,-10){\vector(1,1){5}}\put(55,-5){\line(1,1){5}}

\put(80,0){\vector(1,-1){5}}\put(85,-5){\line(1,-1){5}}
\put(80,0){\vector(1,1){5}}\put(85,5){\line(1,1){5}}

\put(60,0){\vector(1,0){10}}\put(70,0){\line(1,0){10}}

\put(68,-5){$\scriptstyle {\bf v}_k$}

\put(51,12){$\scriptstyle  {\bf v}_i$}
\put(51,-12){$\scriptstyle  {\bf v}_j$}

\put(90,12){$\scriptstyle  {\bf v}_{j'}$}
\put(90,-12){$\scriptstyle  {\bf v}_{i'}$}

\put(66,4){$\scriptstyle  \lambda_k, \,\,Q_k$}

\put(-6,-18){\small vertex}
\put(66,-18){\small edge}

\put(-4,-24){\small (a)}
\put(68,-24){\small (b)}

\end{picture}
\caption{Gluing rules.}
\label{gluing}
\end{center}
\end{figure}

\subsection{Check of gluing rules}


Consider any intertwining operator of the $\calU$ modules 
$\calF^{{\bf  v}_1}_{u_1}\otimes \cdots \otimes \calF^{{\bf  v}_m}_{u_m}
\rightarrow  \calF^{{\bf  v}'_1}_{u'_1}\otimes \cdots \otimes \calF^{{\bf  v}'_n}_{u'_n}$
obtained by composing the trivalent intertwining operators $\Phi$ and $\Phi^*$ in a certain way.
The matrix elements 
can be evaluated by virtue of Proposition \ref{mat-el-IKV} or Proposition \ref{mat-el}. Then we 
have a (not necessarily connected) graph with trivalent vertices, with the following structure associated:
\begin{enumerate}
\item a spectral parameter and a vector $\in \bbZ^2$ is attached to each edge,
\item the condition (\ref{sum}) is satisfied with respect to every vertex,
\item to eace vertex a refined topological vertex is associated (Propositions \ref{mat-el-IKV}, \ref{mat-el}), 
\item each internal edge gives a contraction of refined topological vertices.
\end{enumerate}

Hence if it is shown that the correct gluing factor (\ref{factor}) 
appears to every internal edge in the matrix element, our 
approach from the representation theory of the algebra $\calU$ precisely
reproduces the quantity derived from the refined topological vertex, 
up to a factor depending on the data attached to the external edges.
One can show that this is the case by checking it 
for all the possible (local) processes stated in Propositions 
\ref{gl-pro-1}, \ref{gl-pro-2}, \ref{gl-pro-3}, \ref{gl-pro-4}, \ref{gl-pro-5} below. 
We demonstrate these for the case 
of Awata-Kanno construction, since our notation gets a little simpler. The calculations for 
the topological vertex of Iqbal-Kozcaz-Vafa goes exactly the same way, and we omit them.

\begin{thm}\label{equiv}
Suppose we choose the preferred direction to be vertical $(0,1)$ in the web diagram.
The matrix element of the composition of the  trivalent intertwining operators $\Phi$ and $\Phi^*$,
and the corresponding quantity derived from the theory of the 
refined topological vertex of Iqbal-Kozcaz-Vafa or Awata-Kanno coincide, up to a factor depending on the data attached to external edges.
\end{thm}


\begin{prp}\label{gl-pro-1}
The matrix element of the composition (see Fig \ref{case-1})
\begin{align}
&
\calF^{(1,L)}_{u}\otimes \calF^{(1,M)}_{v}
\mathop{\longrightarrow}^{\Phi^*\otimes {\rm id}}
\calF^{(1,L-1)}_{u/w}\otimes \calF^{(0,1)}_{-w} \otimes \calF^{(1,M)}_{v}
 \mathop{\longrightarrow}^{{\rm id}\otimes \Phi}
\calF^{(1,L-1)}_{u/w}\otimes  \calF^{(1,M+1)}_{vw},\nonumber
\end{align}
with respect to 
$\bra{\iota P_\nu\otimes \iota P_\sigma} $
and $\ket{\iota Q_\mu\otimes \iota Q_\rho}$
is 
\begin{align}
&(-t^{-1/2}w)^{-|\mu|+|\nu|-|\rho|+|\sigma|} f_\nu f_\rho^{-1} 
\sum_\lambda 
\left(w^{L-M} v/u \right)^{|\lambda|} f_\lambda^{L-M-1}
{C_{\mu\lambda}}^{\nu} {C^{\sigma\lambda}}_{\rho}. \nonumber
\end{align}
Recall we should reverse the vertical arrow, and apply the rule for 
calculating the gluing factor (\ref{factor}).
We have $(1,M)\wedge (1,L-1)=L-M-1$, and thus
the factor $\left(w^{L-M} v/u \right)^{|\lambda|} f_\lambda^{L-M-1}$
agrees with the gluing factor  (\ref{factor}).
\end{prp}

\begin{figure}
\begin{center}
\begin{picture}(60,110)(90,-40)\setlength{\unitlength}{1.mm}

\put(0,0){\vector(0,1){5}}\put(0,5){\line(0,1){5}}

\put(-10,-10){\vector(1,1){5}}\put(-5,-5){\line(1,1){5}}
\put(-10,20){\vector(1,-1){5}}\put(-5,15){\line(1,-1){5}}

\put(0,0){\vector(1,0){10}}\put(10,0){\line(1,0){10}}
\put(0,10){\vector(1,0){10}}\put(10,10){\line(1,0){10}}




\put(-16,22){$\scriptstyle (1,M),v$}
\put(-16,-13){$\scriptstyle (1,L),u$}



\put(12,15){$\scriptstyle (1,M+1),vw$}
\put(12,-5){$\scriptstyle (1,L-1),u/w$}

\put(2,5){$\scriptstyle (0,1),-w$}



\put(0,-5){$\scriptstyle \Phi^*$}
\put(0,13){$\scriptstyle \Phi$}

\put(-11,-7){$\scriptstyle \mu$}
\put(-11,16){$\scriptstyle \rho$}

\put(-3,3){$\scriptstyle \lambda$}

\put(18,7){$\scriptstyle \sigma$}

\put(18,2){$\scriptstyle \nu$}

\end{picture}
\caption{Case 1}\label{case-1}
\end{center}
\end{figure}

\begin{proof}
We have
\begin{align*}
&
\iota Q_\mu \otimes \iota Q_\rho
\mapsto
\sum_{\lambda} \Phi^*_\lambda (\iota Q_\mu) \otimes Q_\lambda  \otimes \iota Q_\rho
\mapsto 
{1\over \langle P_\lambda,P_\lambda\rangle_{q,t}}
\sum_{\lambda} \Phi^*_\lambda (\iota Q_\mu) \otimes \Phi_\lambda( \iota Q_\rho).
\end{align*}
Hence the matrix element is 
\begin{align*}
 &
 \sum_{\lambda} 
 {1\over \langle P_\lambda,P_\lambda\rangle_{q,t}}
 \bra{\iota P_\nu} \Phi^*_\lambda  
 \ket{\iota Q_\mu } 
  \bra{\iota P_\sigma} \Phi_\lambda  
 \ket{\iota Q_\rho } .
\end{align*}

\end{proof}

\begin{prp}\label{gl-pro-2}
The matrix element of the composition (see Fig. \ref{case-2})
\begin{align}
&
\calF^{(0,1)}_{-y}\otimes \calF^{(1,L)}_{u}
\mathop{\longrightarrow}^{{\rm id}\otimes \Phi^*}
\calF^{(0,1)}_{-y}\otimes \calF^{(1,L-1)}_{u/x} \otimes \calF^{(0,1)}_{-x}
 \mathop{\longrightarrow}^{\Phi \otimes {\rm id} }
\calF^{(1,L)}_{uy/x}\otimes  \calF^{(0,1)}_{-x},\nonumber
\end{align}
with respect to 
$\bra{\iota P_\sigma \otimes P_\lambda} $
and $\ket{Q_\rho \otimes \iota Q_\mu}$
is 
\begin{align}
&
\left(q x^{L-1}\over -t^{1/2}u/x \right)^{|\lambda|}f_\lambda^{L-1}
\left( -t^{1/2}u/x  \over q y^{L-1}\right)^{|\rho|}f_\rho^{-L+1}\\
&\times
(-t^{-1/2}x)^{-|\mu|} (-t^{-1/2}y)^{|\sigma|}
 \sum_\nu
(x/y)^{|\nu|}
{C_{\mu\nu}}^{\lambda} {C^{\sigma\rho}}_{\nu}. \nonumber
\end{align}
Note that we have $(1,L)\wedge (1,L)=0$, 
and the factor $(x/y)^{|\nu|}$ agrees with (\ref{factor}).
\end{prp}

\begin{figure}
\begin{center}
\begin{picture}(60,90)(90,-40)\setlength{\unitlength}{1.mm}

\put(0,0){\vector(0,1){5}}\put(0,5){\line(0,1){5}}

\put(-10,-10){\vector(1,1){5}}\put(-5,-5){\line(1,1){5}}

\put(0,0){\vector(1,0){10}}\put(10,0){\line(1,0){10}}

\put(20,0){\vector(1,1){5}}\put(25,5){\line(1,1){5}}

\put(20,-10){\vector(0,1){5}}\put(20,-5){\line(0,1){5}}




\put(-16,-13){$\scriptstyle (1,L),u$}



\put(2,-5){$\scriptstyle (1,L-1),u/x$}
\put(26,12){$\scriptstyle (1,L),uy/x$}

\put(2,8){$\scriptstyle (0,1),-x$}
\put(22,-8){$\scriptstyle (0,1),-y$}



\put(-5,-1){$\scriptstyle \Phi^*$}
\put(22,-1){$\scriptstyle \Phi$}

\put(-11,-7){$\scriptstyle \mu$}
\put(20,-12){$\scriptstyle \rho$}

\put(-3,6){$\scriptstyle \lambda$}

\put(25,8){$\scriptstyle \sigma$}

\put(8,2){$\scriptstyle \nu$}

\end{picture}
\caption{Case 2}\label{case-2}
\end{center}
\end{figure}

\begin{proof}
We have
\begin{align*}
&
 Q_\rho \otimes \iota Q_\mu
\mapsto
\sum_{\nu}  Q_\rho \otimes  \Phi^*_\nu (\iota Q_\mu) \otimes Q_\nu
\mapsto 
{1\over \langle P_\rho,P_\rho\rangle_{q,t}}
\sum_{\nu} 
\Phi_\rho
\Phi^*_\nu (\iota Q_\mu) \otimes Q_\nu.
\end{align*}
Hence the matrix element is 
\begin{align*}
 &
  {1\over \langle P_\rho,P_\rho\rangle_{q,t}}
 \bra{\iota P_\sigma} \Phi_\rho\Phi^*_\lambda  
 \ket{\iota Q_\mu } 
 =
 \sum_{\nu} 
  {1\over \langle P_\rho,P_\rho\rangle_{q,t}}
 \bra{\iota P_\sigma} \Phi_\rho
 \ket{\iota Q_\nu}
 \bra{\iota P_\nu}
 \Phi^*_\lambda  
 \ket{\iota Q_\mu } .
\end{align*}

\end{proof}

\begin{prp}\label{gl-pro-3}
The matrix element of the composition (see Fig. \ref{case-3})
\begin{align}
&
\calF^{(0,1)}_{-x}\otimes \calF^{(1,M)}_{v}
\mathop{\longrightarrow}^{ \Phi}
 \calF^{(1,M+1)}_{vx}
 \mathop{\longrightarrow}^{\Phi^*}
\calF^{(1,M)}_{vx/y}\otimes  \calF^{(0,1)}_{-y},\nonumber
\end{align}
with respect to 
$\bra{\iota P_\sigma \otimes P_\rho} $
and $\ket{Q_\lambda \otimes \iota Q_\mu}$
is 
\begin{align}
&
\left( -t^{1/2}v  \over q x^{M}\right)^{|\lambda|}f_\lambda^{-M}
\left(q y^{M}\over -t^{1/2}vx/y \right)^{|\rho|}f_\rho^{M}\\
&\times
(-t^{-1/2}x)^{-|\mu|} (-t^{-1/2}y)^{|\sigma|}
f_\mu^{-1}f_\sigma 
 \sum_\nu
(x/y)^{|\nu|}
{C_{\nu\rho}}^{\sigma} {C^{\nu\lambda}}{\mu}. \nonumber
\end{align}
Note that we have $(1,M)\wedge (1,M)=0$, and the factor $(x/y)^{|\nu|}$
agrees with (\ref{factor}).
\end{prp}

\begin{figure}
\begin{center}
\begin{picture}(60,80)(90,-40)\setlength{\unitlength}{1.mm}

\put(0,-10){\vector(0,1){5}}\put(0,-5){\line(0,1){5}}

\put(-10,10){\vector(1,-1){5}}\put(-5,5){\line(1,-1){5}}

\put(0,0){\vector(1,0){10}}\put(10,0){\line(1,0){10}}

\put(20,0){\vector(1,-1){5}}\put(25,-5){\line(1,-1){5}}

\put(20,0){\vector(0,1){5}}\put(20,5){\line(0,1){5}}




\put(-16,12){$\scriptstyle (1,M),v$}



\put(2,-5){$\scriptstyle (1,M+1),vx$}
\put(28,-13){$\scriptstyle (1,M),vx/y$}

\put(-13,-8){$\scriptstyle (0,1),-x$}
\put(22,7){$\scriptstyle (0,1),-y$}



\put(-5,-1){$\scriptstyle \Phi$}
\put(22,-1){$\scriptstyle \Phi^*$}

\put(-11,6){$\scriptstyle \mu$}
\put(20,12){$\scriptstyle \rho$}

\put(-3,-13){$\scriptstyle \lambda$}

\put(25,-10){$\scriptstyle \sigma$}

\put(8,2){$\scriptstyle \nu$}

\end{picture}
\caption{Case 3}\label{case-3}
\end{center}
\end{figure}

\begin{proof}
We have
\begin{align*}
&
 Q_\lambda \otimes \iota Q_\mu
\mapsto
{1\over \langle P_\lambda,P_\lambda \rangle_{q,t}}
\Phi_\lambda (\iota Q_\mu)
\mapsto 
{1\over \langle P_\lambda,P_\lambda \rangle_{q,t}}
\sum_{\nu} 
\Phi^*_\nu
\Phi_\lambda (\iota Q_\mu) \otimes Q_\nu.
\end{align*}
Hence the matrix element is 
\begin{align*}
 &
{1\over \langle P_\lambda,P_\lambda \rangle_{q,t}}
 \bra{\iota P_\sigma} \Phi^*_\rho\Phi_\lambda  
 \ket{\iota Q_\mu } 
 =
 \sum_{\nu} 
{1\over \langle P_\lambda,P_\lambda \rangle_{q,t}}
 \bra{\iota P_\sigma} \Phi^*_\rho
 \ket{\iota Q_\nu}
 \bra{\iota P_\nu}
 \Phi_\lambda  
 \ket{\iota Q_\mu } .
\end{align*}

\end{proof}

\begin{prp}\label{gl-pro-4}
The matrix element of the composition (see Fig. \ref{case-4})
\begin{align}
&
\calF^{(0,1)}_{-y}\otimes
\calF^{(0,1)}_{-x}\otimes \calF^{(1,M)}_{v}
\mathop{\longrightarrow}^{{\rm id}\otimes \Phi}
\calF^{(0,1)}_{-y}\otimes \calF^{(1,M+1)}_{vx}
 \mathop{\longrightarrow}^{\Phi}
\calF^{(1,M+2)}_{vxy},\nonumber
\end{align}
with respect to 
$\bra{\iota P_\sigma} $
and $\ket{Q_\rho\otimes Q_\lambda \otimes \iota Q_\mu}$
is 
\begin{align}
&
\left( -t^{1/2}v  \over q x^{M}\right)^{|\lambda|}f_\lambda^{-M}
\left( -t^{1/2}vx  \over q y^{M+1}\right)^{|\rho|}f_\rho^{-M-1}\\
&\times
(-t^{-1/2}x)^{-|\mu|} (-t^{-1/2}y)^{|\sigma|}
f_\mu^{-1}
 \sum_\nu
(x/y)^{|\nu|}f_\nu^{-1}
{C^{\sigma\rho}}_{\nu} {C^{\nu\lambda}}_{\mu}. \nonumber
\end{align}
Note that we have $(1,M)\wedge (0,-1)=-1$,
and the factor $(x/y)^{|\nu|}f_\nu^{-1}$ agrees with (\ref{factor}).
\end{prp}

\begin{figure}
\begin{center}
\begin{picture}(60,80)(90,-40)\setlength{\unitlength}{1.mm}

\put(0,-10){\vector(0,1){5}}\put(0,-5){\line(0,1){5}}

\put(-10,10){\vector(1,-1){5}}\put(-5,5){\line(1,-1){5}}

\put(0,0){\vector(1,0){10}}\put(10,0){\line(1,0){10}}

\put(20,0){\vector(1,1){5}}\put(25,5){\line(1,1){5}}

\put(20,-10){\vector(0,1){5}}\put(20,-5){\line(0,1){5}}




\put(-16,12){$\scriptstyle (1,M),v$}



\put(2,-5){$\scriptstyle (1,M+1),vx$}
\put(28,13){$\scriptstyle (1,M+2),vxy$}

\put(-13,-8){$\scriptstyle (0,1),-x$}
\put(22,-7){$\scriptstyle (0,1),-y$}



\put(-5,-1){$\scriptstyle \Phi$}
\put(22,-1){$\scriptstyle \Phi$}

\put(-11,6){$\scriptstyle \mu$}
\put(20,-12){$\scriptstyle \rho$}

\put(-3,-13){$\scriptstyle \lambda$}

\put(25,10){$\scriptstyle \sigma$}

\put(8,2){$\scriptstyle \nu$}

\end{picture}
\caption{Case 4}\label{case-4}
\end{center}
\end{figure}

\begin{proof}
We have
\begin{align*}
&
 Q_\rho\otimes Q_\lambda \otimes \iota Q_\mu
\mapsto
{1\over \langle P_\lambda,P_\lambda \rangle_{q,t}} Q_\rho\otimes
\Phi_\lambda (\iota Q_\mu)
\mapsto 
{1\over \langle P_\lambda,P_\lambda \rangle_{q,t}}
{1\over \langle P_\rho,P_\rho \rangle_{q,t}}
\Phi_\rho
\Phi_\lambda (\iota Q_\mu) .
\end{align*}
Hence the matrix element is 
\begin{align*}
 &
{1\over \langle P_\lambda,P_\lambda \rangle_{q,t}}
{1\over \langle P_\rho,P_\rho \rangle_{q,t}}
 \bra{\iota P_\sigma} \Phi_\rho\Phi_\lambda  
 \ket{\iota Q_\mu } \\
& =
 \sum_{\nu} 
{1\over \langle P_\rho,P_\rho \rangle_{q,t}}
 \bra{\iota P_\sigma} \Phi_\rho
 \ket{\iota Q_\nu}
 {1\over \langle P_\lambda,P_\lambda \rangle_{q,t}}
 \bra{\iota P_\nu}
 \Phi_\lambda  
 \ket{\iota Q_\mu } .
\end{align*}

\end{proof}

\begin{prp}\label{gl-pro-5}
The matrix element of the composition (see Fig. \ref{case-5})
\begin{align}
&
 \calF^{(1,L)}_{u}
\mathop{\longrightarrow}^{\Phi^*}
 \calF^{(1,L-1)}_{u/x} \otimes \calF^{(0,1)}_{-x}
 \mathop{\longrightarrow}^{\Phi^* \otimes {\rm id} }
\calF^{(1,L-2)}_{u/xy}\otimes  \calF^{(0,1)}_{-y}\otimes  \calF^{(0,1)}_{-x},\nonumber
\end{align}
with respect to 
$\bra{\iota P_\sigma \otimes P_\rho \otimes P_\lambda} $
and $\ket{ \iota Q_\mu}$
is 
\begin{align}
&
\left(q x^{L-1}\over -t^{1/2}u/x \right)^{|\lambda|}f_\lambda^{L-1}
\left(q y^{L-2}\over -t^{1/2}u/xy \right)^{|\rho|}f_\rho^{L-2}\\
&\times
(-t^{-1/2}x)^{-|\mu|} (-t^{-1/2}y)^{|\sigma|}
 f_\sigma
 \sum_\nu 
(x/y)^{|\nu|} f_\nu
{C_{\mu\lambda}}^{\nu} {C_{\nu\rho}}^{\sigma}. \nonumber
\end{align}
Note that we have $(1,L)\wedge (0,1)=1$,
and the factor $(x/y)^{|\nu|} f_\nu$ agrees with (\ref{factor}).
\end{prp}

\begin{figure}
\begin{center}
\begin{picture}(60,80)(90,-40)\setlength{\unitlength}{1.mm}

\put(0,0){\vector(0,1){5}}\put(0,5){\line(0,1){5}}

\put(-10,-10){\vector(1,1){5}}\put(-5,-5){\line(1,1){5}}

\put(0,0){\vector(1,0){10}}\put(10,0){\line(1,0){10}}

\put(20,0){\vector(1,-1){5}}\put(25,-5){\line(1,-1){5}}

\put(20,0){\vector(0,1){5}}\put(20,5){\line(0,1){5}}




\put(-16,-13){$\scriptstyle (1,L),u$}



\put(2,-5){$\scriptstyle (1,L-1),u/x$}
\put(26,-13){$\scriptstyle (1,L-2),u/xy$}

\put(2,8){$\scriptstyle (0,1),-x$}
\put(22,8){$\scriptstyle (0,1),-y$}



\put(-5,-1){$\scriptstyle \Phi^*$}
\put(22,-1){$\scriptstyle \Phi^*$}

\put(-11,-7){$\scriptstyle \mu$}
\put(20,12){$\scriptstyle \rho$}

\put(0,12){$\scriptstyle \lambda$}

\put(25,-9){$\scriptstyle \sigma$}

\put(8,2){$\scriptstyle \nu$}

\end{picture}
\caption{Case 5}\label{case-5}
\end{center}
\end{figure}

\begin{proof}
We have
\begin{align*}
&
\iota Q_\mu
\mapsto
\sum_\nu \Phi^* _\nu(\iota Q_\mu)\otimes Q_\nu
\mapsto 
\sum_{\nu,\rho}\Phi^* _\rho \Phi^* _\nu(\iota Q_\mu)\otimes 
Q_\rho\otimes Q_\nu.
\end{align*}
Hence the matrix element is 
\begin{align*}
 &
 \bra{\iota P_\sigma} \Phi^*_\rho\Phi^*_\lambda  
 \ket{\iota Q_\mu } 
=
 \sum_{\nu} 
 \bra{\iota P_\sigma} \Phi^*_\rho
 \ket{\iota Q_\nu}
 \bra{\iota P_\nu}
 \Phi_\lambda^*  
 \ket{\iota Q_\mu } .
\end{align*}

\end{proof}

\section{Examples of compositions of intertwining operators}
We have shown in Theorem \ref{equiv} that 
our construction based on the intertwining operators $\Phi,\Phi^*$ 
derives the same result as the one from the theory of the refined topological vertex.
Based on the findings in \cite{IKV:2009}, \cite{T:2008}, \cite{AK:2009},
it explains clearly the reason why the Nekrasov partition functions 
appear from matrix elements of intertwining operators of the algebra $\calU$. 
In this section, we try to have an interpretation of the spectral parameters 
attached to our Fock modules 
by looking at two examples of the Nekrasov partition functions
\cite{N:2003}, \cite{FP:2003}.

\subsection{Pure $SU(N_c)$ partition function}

Recall the formula of the instanton part of the  ($K$-theoretic) partition function $Z^{\rm inst}_m$ of 
the pure $SU(N_c)$ gauge theory on $\bbR^4 \times S^1$ 
with eight supercharges, associated with the $m$-th power ${\mathcal L}^{\otimes m}$
of the line bundle ${\mathcal L}$ over the instanton moduli space $M(N_c,k)$
\begin{align}
&
Z^{\rm inst}_m({\bf e}_1,\cdots,{\bf e}_{N_c},\Lambda;q,t)
=
\sum_{\lambda^{(1)},\ldots,\lambda^{(N_c)}}
{
\prod_{\alpha=1}^{N_c} ((q^{1/2}t^{-1/2})^{-N_c} \Lambda^{2N_c} (-{\bf e}_\alpha)^{-m})^{|\lambda^{(\alpha)}|}
f_{\lambda^{(\alpha)}}^{-m} \over 
 \prod_{\alpha,\beta =1}^{N_c}
 N_{\lambda^{(\alpha)},\lambda^{(\beta)}} ({\bf e}_\alpha/{\bf e}_\beta)
},
\end{align}
where  the notation
\begin{align}
N_{\lambda,\mu}(u)
&=
\prod_{(i,j)\in \lambda}(1-u q^{-\mu_i+j-1}t^{-\lambda'_j+i}) \cdot 
\prod_{(k,l)\in \mu}(1-u q^{\lambda_k-l}t^{\mu'_l-k+1}) \label{Nlammu}\\
&=
\prod_{\square\in \lambda}
(1-u q^{-a_\mu(\square)-1}t^{-\ell_\lambda(\square)}) \cdot 
\prod_{\blacksquare \in \mu}
(1-u q^{a_\lambda(\blacksquare)}t^{\ell_\mu(\blacksquare)+1}),
\end{align}
has been used.
We demonstrate how $Z^{\rm inst}_m$ appears from our construction.

Let $L,M\in \bbZ$ and $u,v,w$ be indeterminates.
Consider the four point operator (see Fig. \ref{pure})
\begin{align}
&
\Phi\left[{ (1,L-1),u/w;(1,M+1),v w \atop
 (1,L),u;(1,M),v} \right]:
\calF^{(1,L)}_{u}\otimes \calF^{(1,M)}_{v}
\longrightarrow
\calF^{(1,L-1)}_{u/w}\otimes  \calF^{(1,M+1)}_{vw},
\end{align}
defined by the composition of the trivalent intertwining operators  
(which we already considered in Proposition \ref{case-1})
\begin{align}
&
\calF^{(1,L)}_{u}\otimes \calF^{(1,M)}_{v}
\mathop{\longrightarrow}^{\Phi^*\otimes {\rm id}}
\calF^{(1,L-1)}_{u/w}\otimes \calF^{(0,1)}_{-w} \otimes \calF^{(1,M)}_{v}
 \mathop{\longrightarrow}^{{\rm id}\otimes \Phi}
\calF^{(1,L-1)}_{u/w}\otimes  \calF^{(1,M+1)}_{vw}.
\end{align}

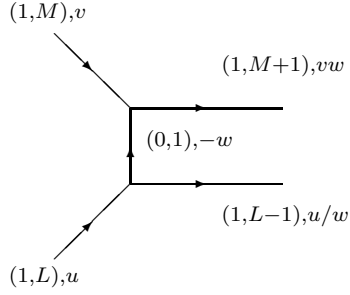
\begin{figure}
\begin{center}
\begin{picture}(60,130)(90,-40)\setlength{\unitlength}{1.mm}

\put(0,0){\vector(0,1){5}}\put(0,5){\line(0,1){5}}

\put(-10,-10){\vector(1,1){5}}\put(-5,-5){\line(1,1){5}}
\put(-10,20){\vector(1,-1){5}}\put(-5,15){\line(1,-1){5}}

\put(0,0){\vector(1,0){10}}\put(10,0){\line(1,0){10}}
\put(0,10){\vector(1,0){10}}\put(10,10){\line(1,0){10}}




\put(-16,22){$\scriptstyle (1,M),v$}
\put(-16,-13){$\scriptstyle (1,L),u$}



\put(12,15){$\scriptstyle (1,M+1),vw$}
\put(12,-5){$\scriptstyle (1,L-1),u/w$}

\put(2,5){$\scriptstyle (0,1),-w$}



\end{picture}
\caption{Four point operator.}
\label{pure}
\end{center}
\end{figure}

For any $\alpha\otimes \beta\in \calF^{(1,L)}_{u}\otimes \calF^{(1,M)}_{v}$, we have
\begin{align}
&
\Phi\left[{ (1,L-1),u/w;(1,M+1),vw \atop
 (1,L),u;(1,M),v} \right](\alpha\otimes \beta)\\
 &=
 \sum_{\lambda}
  {1\over \langle P_{\lambda},P_{\lambda}\rangle_{q,t}}
  \Phi^* _{\lambda}\left[{(1,L-1),u/w;(0,1),-w
 \atop
 (1,L),u }\right]
 (\alpha)
 \otimes
  \Phi_{\lambda}\left[{(1,M+1),vw
 \atop
 (0,1),-w;(1,M),v }\right]
( \beta)\nonumber\\
 &=
 \sum_{\lambda}
 { (q^{-1/2}t^{1/2}u^{-1}v w^{L-M})^{|\lambda|} f_\lambda^{L-M-1}\over N_{\lambda,\lambda}(1) }
\left(
:\tPhis_\emptyset(-w) \xi_{\lambda}(-w): \alpha \right)
\otimes \left(
:\tPhi_\emptyset(-w) \eta_{\lambda}(-w):\beta\right)\nonumber 
\end{align}
{}from Theorems \ref{thm-1}, {\ref{thm-2} and the formula (\ref{ccp}).

Let $w_1,w_2,\ldots,w_{N_c}$ be a set of indeterminates.
Set
\begin{align}
&u_i=u \prod_{k=1}^{i-1} w_k^{-1},\qquad
v_i=v \prod_{k=1}^{i-1} w_k, \qquad (i=1,2,\ldots,N_c+1),
\end{align}
for simplicity. 
Define the four point operator (see Fig. \ref{web-1})
\begin{align}
&
\Phi\left[{ (1,L-N_c),u_{N_c+1};(1,M+N_c),v_{N_c+1} \atop
 (1,L),u_1;(1,M),v_1} \right]:
\calF^{(1,L)}_{u_1}\otimes \calF^{(1,M)}_{v_1}
\longrightarrow
\calF^{(1,L-N_c)}_{u_{N_c+1}}\otimes  \calF^{(1,M+N_c)}_{v_{N_c+1}},
\end{align}
as the composition
\begin{align}
&\Phi\left[{ (1,L-N_c),u_{N_c+1};(1,M+N_c),v_{N_c+1} \atop
(1,L),u_1;(1,M),v_1} \right]\\
 &=
 \Phi\left[{ (1,L-N_c),u_{N_c+1};(1,M+N_c),v_{N_c+1} \atop
(1,L-N_c+1),u_{N_c};(1,M+N_c-1),v_{N_c} } \right] \cdots 
 \Phi\left[{ (1,L-1),u_2;(1,M+1),v_2 \atop
 (1,L),u_1;(1,M),v_1} \right].\nonumber 
\end{align}

\begin{figure}
\begin{center}
\begin{picture}(60,130)(90,-40)\setlength{\unitlength}{1.mm}

\put(0,0){\vector(0,1){5}}\put(0,5){\line(0,1){5}}

\put(-10,-10){\vector(1,1){5}}\put(-5,-5){\line(1,1){5}}
\put(-10,20){\vector(1,-1){5}}\put(-5,15){\line(1,-1){5}}

\put(0,0){\vector(1,0){10}}\put(10,0){\line(1,0){10}}
\put(0,10){\vector(1,0){10}}\put(10,10){\line(1,0){10}}

\put(40,0){\vector(1,0){10}}\put(50,0){\line(1,0){10}}
\put(40,10){\vector(1,0){10}}\put(50,10){\line(1,0){10}}

\put(60,0){\vector(0,1){5}}\put(60,5){\line(0,1){5}}

\put(60,0){\vector(1,-1){5}}\put(65,-5){\line(1,-1){5}}
\put(60,10){\vector(1,1){5}}\put(65,15){\line(1,1){5}}

\put(-16,22){$\scriptstyle (1,M),v_1$}
\put(-16,-13){$\scriptstyle (1,L),u_1$}

\put(4,17){$\scriptstyle (1,M+1),v_2$}
\put(4,-8){$\scriptstyle (1,L-1),u_2$}

\put(34,17){$\scriptstyle (1,M+N_c-1),v_{N_c}$}
\put(34,-8){$\scriptstyle (1,L-N_c+1),u_{N_c}$}

\put(69,22){$\scriptstyle (1,M+N_c),v_{N_c+1}$}
\put(69,-13){$\scriptstyle (1,L-N_c),u_{N_c+1}$}

\put(2,5){$\scriptstyle (0,1),-w_1$}

\put(62,5){$\scriptstyle (0,1),-w_{N_c}$}

\put(27,-1){$\cdots$}
\put(27,9){$\cdots$}

\end{picture}
\caption{Web diagram for pure $SU(N)$ gauge partition function.}
\label{web-1}
\end{center}
\end{figure}
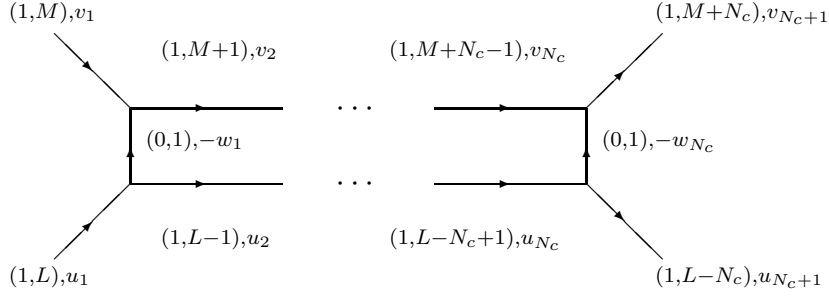

\begin{prp}\label{propos-1}
When we identify parameters as
\begin{align}
{\bf e}_i=-w_i,\qquad 
\Lambda^{2N_c}= {v\over u}\prod_{i=1}^{N_c} w_i,\qquad 
m=-L+M+N_c,
\end{align}
we have
\begin{align}
&
\bra{P_\emptyset  \otimes P_\emptyset }
\Phi\left[{ (1,L-N_c),u_{N_c+1};(1,M+N_c),v_{N_c+1} \atop
 (1,L),u_1;(1,M),v_1} \right]\ket{P_\emptyset  \otimes P_\emptyset }\nonumber\\
 &=
 \prod_{1\leq i<j\leq N}
{\mathcal G}({\bf e}_i/{\bf e}_j){\mathcal G}(q t^{-1}{\bf e}_i/{\bf e}_j)
\cdot 
Z^{\rm inst}_m({\bf e}_1,\cdots,{\bf e}_{N_c},\Lambda;q,t), \label{pro-1}
\end{align}
where
\begin{align}
{\mathcal G}(u)=
\exp\left(-
\sum_{n>0}{1\over n}{1\over (1-q^n)(1-t^{-n})} u^n
\right)\in \bbQ(q,t)[[u]]. \label{calG}
\end{align}
\end{prp}

A proof of this will be given in Section \ref{section-7}.


\subsection{$SU(N_c)$ with $N_f=2 N_c$.}
Next, we turn to the case with $N_f=2 N_c$ fundamental matters.
Let $u,v,w,x,y$ be indeterminates.
Consider the six point operator
\begin{align}
&
\Phi\left[{ (1,0),ux/w;(1,0),v w/y;(0,1),-y \atop
(0,1),-x; (1,0),u;(1,0),v} \right]\\
 &:
\calF^{(0,1)}_{-x}\otimes
\calF^{(1,0)}_{u}\otimes \calF^{(1,0)}_{v}
\longrightarrow
\calF^{(1,0)}_{ux/w}\otimes  \calF^{(1,0)}_{vw/y}\otimes \calF^{(0,1)}_{-y},\nonumber
\end{align}
defined by the composition of the intertwining operators
\begin{align}
&
\calF^{(0,1)}_{-x}\otimes
\calF^{(1,L)}_{u}\otimes \calF^{(1,M)}_{v}
\mathop{\longrightarrow}^{{\rm id}\otimes\Phi^*\otimes {\rm id}}
\calF^{(0,1)}_{-x}\otimes
\calF^{(1,L-1)}_{u/w}\otimes \calF^{(0,1)}_{-w} \otimes \calF^{(1,M)}_{v}\\
&
 \mathop{\longrightarrow}^{{\rm id}^{\otimes 2}\otimes \Phi}
 \calF^{(0,1)}_{-x}\otimes
\calF^{(1,L-1)}_{u/w}\otimes  \calF^{(1,M+1)}_{vw}
 \mathop{\longrightarrow}^{ \Phi\otimes \Phi^*}
 \calF^{(1,L)}_{ux/w}\otimes  \calF^{(1,M)}_{vw/y}\otimes  \calF^{(0,1)}_{-y}.\nonumber
\end{align}

\begin{figure}
\begin{center}
\begin{picture}(60,180)(80,-70)\setlength{\unitlength}{1.mm}

\put(0,0){\vector(0,1){5}}\put(0,5){\line(0,1){5}}

\put(-10,0){\vector(1,0){5}}\put(-5,0){\line(1,0){5}}
\put(-10,10){\vector(1,0){5}}\put(-5,10){\line(1,0){5}}

\put(0,0){\vector(1,-1){5}}\put(5,-5){\line(1,-1){5}}
\put(0,10){\vector(1,1){5}}\put(5,15){\line(1,1){5}}

\put(10,-10){\vector(1,0){5}}\put(15,-10){\line(1,0){5}}
\put(10,20){\vector(1,0){5}}\put(15,20){\line(1,0){5}}

\put(10,-20){\vector(0,1){5}}\put(10,-15){\line(0,1){5}}
\put(10,20){\vector(0,1){5}}\put(10,25){\line(0,1){5}}




\put(-20,12){$\scriptstyle (1,0),v$}
\put(-20,-3){$\scriptstyle (1,0),u$}

\put(-7,17){$\scriptstyle (1,1),vw$}
\put(-8,-9){$\scriptstyle (1,-1),u/w$}



\put(18,16){$\scriptstyle (1,0),vw/y$}
\put(18,-8){$\scriptstyle (1,0),ux/w$}

\put(2,5){$\scriptstyle (0,1),-w$}

\put(12,29){$\scriptstyle (0,1),-y$}

\put(12,-21){$\scriptstyle (0,1),-x$}



\end{picture}
\caption{Six point operator.}
\end{center}
\end{figure}
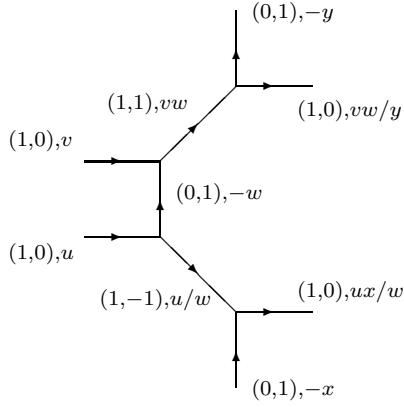

For any $P_\lambda\otimes \alpha\otimes \beta \in
\calF^{(0,1)}_{-x}\otimes \calF^{(1,0)}_u\otimes \calF^{(1,0)}_v$,
we have
\begin{align}
&
\Phi\left[{ (1,0),ux/w;(1,0),v w/y;(0,1),-y \atop
(0,1),-x; (1,0),u;(1,0),v} \right](P_\lambda\otimes \alpha\otimes \beta)\nonumber\\
 &=
 \sum_{\mu,\nu}
  {1\over \langle P_{\mu},P_{\mu}\rangle_{q,t}}\nonumber \\
  &\times 
    \Phi_{\lambda}
    \left[{(1,0),ux/w
 \atop
 (0,1),-x;(1,-1),u/w }\right]
 \Phi^* _{\mu}\left[{(1,-1),u/w;(0,1),-w
 \atop
 (1,0),u }\right]
 (\alpha)\\
 &
 \otimes
  \Phi^* _{\nu}\left[{(1,0),vw/y;(0,1),-y
 \atop
 (1,1),vw }\right]
   \Phi_{\mu}\left[{(1,1),vw
 \atop
 (0,1),-w;(1,0),v }\right]
( \beta)\otimes Q_\nu\nonumber\\
 &=
 \sum_{\mu,\nu}
    {     q^{n(\lambda')} (ux/w)^{|\lambda|}
   q^{n(\nu')} (qy/vw)^{|\nu|} 
\over c_\lambda c_\nu }  
   {    
   (q^{1/2}t^{-1/2})^{-|\mu|}(v/u)^{|\mu|}f_\mu^{-1}
\over N_{\mu,\mu}(1) }  
  \nonumber \\
 &
\times \left(
:\tPhi_\emptyset(-x) \eta_{\lambda}(-x):
:\tPhis_\emptyset(-w) \xi_{\mu}(-w): 
 \alpha\right)\nonumber\\
&
\otimes \left(
:\tPhis_\emptyset(-y) \xi_{\nu}(-y): 
:\tPhi_\emptyset(-w) \eta_{\mu}(-w):
 \beta\right)\otimes Q_\nu,\nonumber
\end{align}

Restricting this six point operator, introduce the four point operator
\begin{align}
&
\Phi\left[{ (1,0),ux/w;(1,0),v w/y \atop
(1,0),u;(1,0),v} \right]
 :
\calF^{(1,0)}_{u}\otimes \calF^{(1,0)}_{v}
\longrightarrow
\calF^{(1,0)}_{ux/w}\otimes  \calF^{(1,0)}_{vw/y},\nonumber
\end{align}
defined by specifying the action
on any $\alpha\otimes \beta \in
 \calF^{(1,0)}_u\otimes \calF^{(1,0)}_v$ as
\begin{align}
&\Phi\left[{ (1,0),ux/w;(1,0),v w/y \atop
(1,0),u;(1,0),v} \right](\alpha\otimes\beta)\\
&=
{\rm id}\otimes {\rm id}\otimes \langle P_\emptyset,  \bullet \rangle _{q,t}\circ
\Phi\left[{ (1,0),ux/w;(1,0),v w/y;(0,1),-y \atop
(0,1),-x; (1,0),u;(1,0),v} \right](P_\emptyset \otimes \alpha\otimes \beta) \nonumber\\
 &=
 \sum_{\mu}
    {    
   (q^{1/2}t^{-1/2})^{-|\mu|}(v/u)^{|\mu|}f_\mu^{-1}
\over N_{\mu,\mu}(1) }  
  \nonumber \\
 &
\times \left(
\tPhi_\emptyset(-x) 
:\tPhis_\emptyset(-w) \xi_{\mu}(-w): 
 \alpha \right)
\otimes \left(
\tPhis_\emptyset(-y)
:\tPhi_\emptyset(-w) \eta_{\mu}(-w):
 \beta\right),\nonumber
\end{align}
where we have used the shorthand notation $ \langle P_\emptyset,  \bullet \rangle _{q,t} \circ Q_\nu=
\langle P_\emptyset,  Q_\nu\rangle _{q,t} $.

Let $u,v,w_1,\ldots,w_{N_c},x_1,\ldots,x_{N_c},y_1,\ldots,y_{N_c}$ be 
a set of indeterminates.
Set 
\begin{align}
u_i=u \prod_{k=1}^{i-1} x_k/w_k,\qquad v_i=v \prod_{k=1}^{i-1} w_k/y_k,
\qquad (i=1,2,\ldots,N_c+1),
\end{align}
for simplicity.

\begin{prp}\label{propos-2}
Set 
\begin{align}
{\bf e}_i=-w_i,\quad 
{\bf e}'_i=-q^{1/2}t^{-1/2}y_i,\quad 
{\bf e}''_i=-q^{-1/2}t^{1/2}x_i,\quad 
\Lambda^{2N_c}= (q^{1/2}t^{-1/2})^{-N}\, {v\over u}\prod_{i=1}^{N_c} {w_i\over y_i}.
\end{align}
We have
\begin{align}
&
\bra{P_\emptyset \otimes P_\emptyset}
\Phi\left[{ (1,0),u_{N_c+1};(1,0),v_{N_c+1} \atop
(1,0),u_{N_c};(1,0),v_{N_c}} \right]
\cdots 
\Phi\left[{ (1,0),u_2;(1,0),v_2 \atop
(1,0),u_1;(1,0),v_1} \right]
\ket{P_\emptyset \otimes P_\emptyset}\label{pro-2}\\
&=
 \prod_{k=1}^{N_c}
{1\over 
{\mathcal G}({\bf e}_k/{\bf e}''_k){\mathcal G}(qt^{-1}{\bf e}_k/{\bf e}'_k)}\cdot
 \prod_{1\leq i<j\leq N_c}
{
 {\mathcal G}({\bf e}_i/{\bf e}_j){\mathcal G}(qt^{-1} {\bf e}_i/{\bf e}_j)
 {\mathcal G}(q t^{-1}{\bf e}''_i/{\bf e}''_j){\mathcal G}( {\bf e}'_i/{\bf e}'_j)
 \over 
{\mathcal G}({\bf e}_i/{\bf e}''_j){\mathcal G}(qt^{-1}{\bf e}''_i/{\bf e}_j)
{\mathcal G}(qt^{-1} {\bf e}_i/{\bf e}'_j){\mathcal G}({\bf e}'_i/{\bf e}_j)}\nonumber \\
&
\times
\sum_{\lambda^{(1)},\cdots ,\lambda^{(N_c)}}
\prod_{k=1}^{N_c}
\Lambda^{2N_c|{\lambda^{(k)}}|}
\prod_{1\leq i,j,\leq N_c}
{N_{\emptyset,\lambda^{(j)}}({\bf e}'_i/{\bf e}_j)
N_{\lambda^{(i)},\emptyset}({\bf e}_i/{\bf e}''_j)
\over N_{\lambda^{(i)},\lambda^{(j)}}({\bf e}_i/{\bf e}_j)}. \nonumber
\end{align}
\end{prp}

Since our proof of this goes in a parallel way as the one for  Proposition \ref{propos-1}, we omit it.



\section{Proofs of Theorems \ref{thm-1} and \ref{thm-2}}\label{section-6}
\subsection{Some formulas concerning $A^\pm_{\lambda,i},B^\pm_\lambda(z)$}

\begin{lem}
Let 
$c_\lambda,c'_\lambda,A^+_{\lambda,i},A^-_{\lambda,i}$ be as in (\ref{c-lam}),
(\ref{A+}), (\ref{A-}). 
We have
\begin{align}
&
{c'_{\lambda+{\bf 1}_i} \over c'_\lambda}
{c_\lambda \over c_{\lambda+{\bf 1}_i}} A^+_{\lambda,i}=-
q A^-_{\lambda+{\bf 1}_i,i},\qquad 
{c'_{\lambda-{\bf 1}_i} \over c'_\lambda}
{c_\lambda \over c_{\lambda-{\bf 1}_i}} A^-_{\lambda,i}=-
q^{-1}A^+_{\lambda-{\bf 1}_i,i}. \label{ccA=}
\end{align}
Hence the action of $\calU$ is written in terms of the basis $(Q_\lambda)$ as
\begin{align}
&
\gamma Q_\lambda=Q_\lambda,\\
&
x^+(z) Q_\lambda=-
\sum_{i=1}^{\ell(\lambda)+1} 
q A^-_{\lambda+{\bf 1}_i,i}\,
\delta(q^{\lambda_i}t^{-i+1}u/z) 
Q_{\lambda+{\bf 1}_i},\\
&
x^-(z) Q_\lambda=
-q^{1/2}t^{-1/2}
\sum_{i=1}^{\ell(\lambda)} 
q^{-1} A^+_{\lambda-{\bf 1}_i,i}\,
\delta(q^{\lambda_i-1}t^{-i+1}u/z)
Q_{\lambda-{\bf 1}_i},\\
&
\psi^+(z)Q_\lambda=
q^{1/2}t^{-1/2}
B^+_\lambda(u/z)Q_\lambda,\\
&
\psi^-(z)Q_\lambda=
q^{-1/2}t^{1/2} \,B^-_\lambda(z/u)Q_\lambda.
\end{align}
\end{lem}

\begin{proof}
{}From the definitions of  $c_\lambda,c'_\lambda$, it immediately follows that
\begin{align*}
{c_{\lambda+{\bf 1}_k} \over c_\lambda}
%
&=
(1-q^{\lambda_k}t^{\ell(\lambda)-k+1})
\prod_{i=1}^{k-1}
{1-q^{\lambda_i-\lambda_k-1}t^{k-i+1}\over 1-q^{\lambda_i-\lambda_k-1}t^{k-i}}
\prod_{j=k+1}^{\ell(\lambda)}
{1-q^{\lambda_k-\lambda_{j}}t^{j-k}\over 1-q^{\lambda_k-\lambda_{j}}t^{j-k+1}},\\
%
{c'_{\lambda+{\bf 1}_k} \over c'_\lambda}
%
&=
(1-q^{\lambda_k+1}t^{\ell(\lambda)-k})
\prod_{i=1}^{k-1}
{1-q^{\lambda_i-\lambda_k}t^{k-i}\over 1-q^{\lambda_i-\lambda_k}t^{k-i-1}}
\prod_{j=k+1}^{\ell(\lambda)}
{1-q^{\lambda_k-\lambda_{j}+1}t^{j-k-1}\over 1-q^{\lambda_k-\lambda_{j}+1}t^{j-k}}.\\
\end{align*}
Noting that 
\begin{align*}
A^-_{\lambda,k}=
(1-t^{-1})
{
1-q^{-\lambda_k}t^{-\ell(\lambda)+k} \over 
1-q^{-\lambda_k+1}t^{-\ell(\lambda)+k-1}}
\prod_{j=k+1}^{\ell(\lambda)}
{1-q^{-\lambda_k+\lambda_{j}}t^{-j+k+1}\over 1-q^{-\lambda_k+\lambda_{j}}t^{-j+k}}
{1-q^{-\lambda_k+\lambda_{j}+1}t^{-j+k-1}\over 
1-q^{-\lambda_k+\lambda_{j}+1}t^{-j+k}},
\end{align*}
one obtains (\ref{ccA=}).
\end{proof}


\begin{lem}\label{phi+-eta2}
Let $B^\pm_\lambda(z)$  be as in (\ref{B+}), (\ref{B-}). 
We have
\begin{align}
&\prod_{i=1}^{\ell(\lambda)}\prod_{j=1}^{\lambda_i}
g(q^{j-1}t^{-i+1}v/z)^{-1} =
{1-v/z\over 1-q^{-1}tv/z}B^+(v/z), \label{*shiki-1}\\
&
\prod_{i=1}^{\ell(\lambda)}\prod_{j=1}^{\lambda_i}
g(q^{-j+1}t^{i-1}v/z) =
{1-z/v\over 
1-qt^{-1}z/v}
B^-_\lambda(z/v), \label{*shiki-2}
\end{align}
where 
$g(z)$ is given in (\ref{g}).
\end{lem}


The following will also be needed.
\begin{lem}\label{eta-etalam2}
We have
\begin{align}
&
\prod_{i=1}^{\ell(\lambda)}\prod_{j=1}^{\lambda_i}
f(q^{j-1}t^{-i+1}v/z)
={1-v/z\over 1-t^{-\ell(\lambda)} v/z} 
\prod_{i=1}^{\ell(\lambda)}
{1-q^{\lambda_i}t^{-i}v/z \over 1-q^{\lambda_i} t^{-i+1}v/z},\\
&
\prod_{i=1}^{\ell(\lambda)}\prod_{j=1}^{\lambda_i}
f(q^{-j+1}t^{i-1}z/v)
={1-q t^{-1}z/v\over 1-qt^{\ell(\lambda-1)} z/v} 
\prod_{i=1}^{\ell(\lambda)}
{1-q^{-\lambda_i+1}t^{i-1}z/v \over 1-q^{-\lambda_i+1} t^{i-2}z/v}.
\end{align}
where 
\begin{align}
&f(z)={(1-z)(1-qt^{-1} z)\over (1-qz)(1-t^{-1} z)}.
\end{align}
\end{lem}

\begin{lem}\label{delta}
We have
\begin{align}
&
{q^{n(\lambda')}\over c_\lambda}
\left(
{1\over 1-t^{-\ell(\lambda)} v/z} 
\prod_{i=1}^{\ell(\lambda)}
{1-q^{\lambda_i}t^{-i}v/z \over 1-q^{\lambda_i} t^{-i+1}v/z}+
{z\over v} 
{1\over 1-t^{\ell(\lambda)} z/v} 
\prod_{i=1}^{\ell(\lambda)}
{1-q^{-\lambda_i}t^{i}z/v \over 1-q^{-\lambda_i} t^{i-1}z/v}
 \right)\\
 &=
 \sum_{i=1}^{\ell(\lambda)+1}
 {q^{n((\lambda+{\bf 1}_i)')}\over c_{\lambda+{\bf 1}_i}}
 A^+_{\lambda,i} \, \delta(q^{\lambda_i}t^{-i+1}v/z).\nonumber
\end{align}
\end{lem}
\begin{proof}
It follows from
\begin{align*}
&
{1\over 1-t^{-\ell(\lambda)} v/z} 
\prod_{i=1}^{\ell(\lambda)}
{1-q^{\lambda_i}t^{-i}v/z \over 1-q^{\lambda_i} t^{-i+1}v/z}+
{z\over v} 
{1\over 1-t^{\ell(\lambda)} z/v} 
\prod_{i=1}^{\ell(\lambda)}
{1-q^{-\lambda_i}t^{i}z/v \over 1-q^{-\lambda_i} t^{i-1}z/v}\\
&
=
\sum_{i=1}^{\ell(\lambda)+1}
q^{\lambda_i}t^{-i+1}
\delta(q^{\lambda_i}t^{-i+1}v/z) 
{1-t\over 1-q^{\lambda_i} t^{\ell(\lambda)-i+1}}
\prod_{j=1}^{i-1}
{1-q^{\lambda_i-\lambda_j}t^{j-i+1} \over 1-q^{\lambda_i-\lambda_j}t^{j-i}} 
\prod_{j=i+1}^{\ell(\lambda)}
{1-q^{\lambda_i-\lambda_j}t^{j-i+1} \over 1-q^{\lambda_i-\lambda_j}t^{j-i}} ,
\end{align*}
and 
\begin{align*}
&
{c_{\lambda+{\bf 1}_i} \over c_\lambda}=
t^{i-1}
(1-q^{\lambda_i}t^{\ell(\lambda)-i+1})
\prod_{j=1}^{i-1}
{1-q^{\lambda_i-\lambda_j+1}t^{j-i-1} \over 
1-q^{\lambda_i-\lambda_j+1}t^{j-i} }
\prod_{j=i+1}^{\ell(\lambda)}
{1-q^{\lambda_i-\lambda_j}t^{j-i}\over 1-q^{\lambda_i-\lambda_j}t^{j-i+1}},\\
&
n(\lambda')=\sum_{i\geq 0} {\lambda_i(\lambda_i-1)\over 2},\qquad 
{q^{n((\lambda+{\bf 1}_i)')} \over q^{n(\lambda')}}=q^{\lambda_i}.
\end{align*}

\end{proof}

\begin{lem}\label{delta-2}
We have 
\begin{align}
&
{q^{n(\lambda')} \over c_\lambda}
\Biggl(
(1-q^{-1} t^{-\ell(\lambda)+1}v/z)
\prod_{i=1}^{\ell(\lambda)}
{1-q^{\lambda_i-1}t^{-i+2}v/z \over 1-q^{\lambda_i-1} t^{-i+1}v/z}\\
&\qquad \qquad 
+
q^{-1}t {v\over z}
(1-q t^{\ell(\lambda)-1}z/v)
\prod_{i=1}^{\ell(\lambda)}
{1-q^{-\lambda_i+1}t^{i-2}z/v \over 1-q^{-\lambda_i+1} t^{i-1}z/v}\Biggr)\nonumber\\
&=
\sum_{i=1}^{\ell(\lambda)}
{q^{n((\lambda-{\bf 1}_i)')} \over c_{\lambda-{\bf 1}_i}}A^-_{\lambda,i}
\delta(q^{\lambda_i-1}t^{-i+1}v/z).\nonumber
\end{align}
\end{lem}

\subsection{Operator product formulas for $\tPhi_\lambda(v)$}

\begin{lem}\label{Phi0}
The operator product formulas between $\tPhi_\emptyset(v)$
and the generators of $\calU$ are
\begin{align}
&
\eta(z) \tPhi_\emptyset (v)={1\over 1-v/z} :\eta(z) \tPhi_\emptyset (v):,\label{Phi0-1}\\
&
\tPhi_\emptyset (v)\eta(z) ={1\over 1-qt^{-1}z/v} :\eta(z) \tPhi_\emptyset (v):,\label{Phi0-2}\\
&
\xi(z) \tPhi_\emptyset (v)=(1-q^{-1/2}t^{1/2}v/z) :\xi(z) \tPhi_\emptyset (v):,\label{Phi0-3}\\
&
\tPhi_\emptyset (v)\xi(z) = (1-q^{1/2}t^{-1/2}z/v) :\xi(z) \tPhi_\emptyset (v):,\label{Phi0-4}\\
&
\varphi^+(q^{1/4}t^{-1/4}z) \tPhi_\emptyset (v)
={1-q^{-1}t v/z \over 1- u/z}  \tPhi_\emptyset (v)\varphi^+(q^{1/4}t^{-1/4}z) ,\label{Phi0-5}\\
&
\varphi^-(q^{-1/4}t^{1/4}z) \tPhi_\emptyset (v)
={1-q t^{-1}z/v \over 1-z/v}  \tPhi_\emptyset (v)\varphi^-(q^{-1/4}t^{1/4}z) .\label{Phi0-6}
\end{align}
\end{lem}

\begin{prp}\label{prop-1}
We have
\begin{align}
&
\varphi^+(q^{1/4}t^{-1/4}z) \tPhi_\lambda(v) \varphi^+(q^{1/4}t^{-1/4}z)^{-1}
=B^+_\lambda(v/z)  \tPhi_\lambda(v),\label{prop-1-1}\\
&
\varphi^-(q^{-1/4}t^{1/4}z) \tPhi_\lambda(v) \varphi^-(q^{-1/4}t^{1/4}z)^{-1}
=B^-_\lambda(z/v)  \tPhi_\lambda(v).\label{prop-1-2}
\end{align}
\end{prp}
\begin{proof}
Note that 
\begin{align*}
&\varphi^+(q^{1/4}t^{-1/4}z) \eta(v) \varphi^+(q^{1/4}t^{-1/4}z)^{-1}
=g(v/z)^{-1}
\eta(v) ,\\
&
\varphi^-(q^{-1/4}t^{1/4}z) \eta(v) \varphi^-(q^{-1/4}t^{1/4}z)^{-1}
=
g(z/v)
 \eta(v) .
\end{align*}
Then (\ref{prop-1-1}), (\ref{prop-1-2}) follow from Lemma \ref{phi+-eta2} and 
(\ref{Phi0-5}), (\ref{Phi0-6}) in Lemma \ref{Phi0}.
\end{proof}

\begin{lem}\label{eta-phi}
We have
\begin{align}
&
\eta(z)\tPhi_\lambda(v)
=
{1\over 1-t^{-\ell(\lambda)} v/z}
\prod_{i=1}^{\ell(\lambda)}
{1-q^{\lambda_i} t^{-i} v/z \over 1-q^{\lambda_i} t^{-i+1} v/z} :\eta(z)\tPhi_\lambda(v):,\label{eta-phi-1}\\
&
 \tPhi_\lambda(v)\eta(z)
=
{1\over 1-q t^{\ell(\lambda)-1} z/v}
\prod_{i=1}^{\ell(\lambda)}
{1-q^{-\lambda_i+1} t^{i-1} z/v \over 1-q^{-\lambda_i+1} t^{i-2} z/v} :\eta(z)\tPhi_\lambda(v):,\label{eta-phi-2}\\
&
B^-_\lambda(z/v)  \tPhi_\lambda (v)\eta(z)
=
{1\over 1-t^{\ell(\lambda)} z/v}
\prod_{i=1}^{\ell(\lambda)}
{1-q^{-\lambda_i} t^{i} z/v \over 1-q^{-\lambda_i} t^{i-1} z/v} :\eta(z)\tPhi_\lambda(v):.\label{eta-phi-3}
\end{align}
\end{lem}
\begin{proof}
We have
$\eta(z)\eta(v)=f(v/z):\eta(z)\eta(v):$. 
Hence from Lemma \ref{eta-etalam2}
\begin{align*}
&
\eta(z) \eta_\lambda(v)=
{1-v/z\over 1-t^{-\ell(\lambda)} v/z} 
\prod_{i=1}^{\ell(\lambda)}
{1-q^{\lambda_i}t^{-i}v/z \over 1-q^{\lambda_i} t^{-i+1}v/z}:\eta(z) \eta_\lambda(v):,\\
&
 \eta_\lambda(v)\eta(z)=
 {1-q t^{-1}z/v\over 1-q t^{\ell(\lambda)-1} z/v} 
\prod_{i=1}^{\ell(\lambda)}
{1-q^{-\lambda_i+1}t^{i-1}z/v \over 1-q^{-\lambda_i+1} t^{i-2}z/v}:\eta(z) \eta_\lambda(v):.
\end{align*}
Then (\ref{eta-phi-1}), (\ref{eta-phi-2}) follow from (\ref{Phi0-1}), (\ref{Phi0-2}) in
Lemma \ref{Phi0}.
\end{proof}

\begin{prp}\label{prop-2}
We have
\begin{align}
\eta(z)\tPhi_\lambda(v)+{z\over v} B^-_\lambda(z/v)\tPhi_\lambda(v) \eta(z)
=
\sum_{i=1}^{\ell(\lambda)+1} A^+_{\lambda,i} \, 
\tPhi_{\lambda+{\bf 1}_i}(v) \delta(q^{\lambda_i}t^{-i+1}v/z).
\end{align}

\end{prp}

\begin{proof}
It follows from  Lemma \ref{delta} and  (\ref{eta-phi-1}), (\ref{eta-phi-3}) in Lemma \ref{eta-phi}.
\end{proof}


\begin{lem}\label{xi-phi}
We have
\begin{align}
&
\xi(q^{1/2}t^{-1/2}z)  \tPhi_\lambda (v)
=(1-q^{-1}t^{-\ell(\lambda)+1}v/z)
\prod_{i=1}^{\ell(\lambda)} 
{1-q^{\lambda_i-1}t^{-i+2}v/z \over 1-q^{\lambda_i-1}t^{-i+1}v/z }
:\xi(q^{1/2}t^{-1/2}z)  \tPhi_\lambda (v):,\label{xi-phi-1}\\
&
\tPhi_\lambda (v)\xi(q^{1/2}t^{-1/2}z)  
=(1-q t^{\ell(\lambda)-1}z/v)
\prod_{i=1}^{\ell(\lambda)} 
{1-q^{-\lambda_i+1}t^{i-2}z/v \over 1-q^{-\lambda_i+1}t^{i-1}z/v }
:\tPhi_\lambda (v)\xi(q^{1/2}t^{-1/2}z)  :.\label{xi-phi-2}
\end{align}
\end{lem} 
\begin{proof}
Note that
$\xi(q^{1/2}t^{-1/2}z)\eta(v)=
f(q^{-1}t v/z)^{-1}
:\xi(q^{1/2}t^{-1/2}z)\eta(v):$, and 
$\eta(v) \xi(q^{1/2}t^{-1/2}z)=
f(z/v)^{-1}
:\xi(q^{1/2}t^{-1/2}z)\eta(v):$. 
Thus from Lemma \ref{eta-etalam2}, we have
\begin{align*}
&
\xi(q^{1/2}t^{-1/2}z) \eta_\lambda(v)
={1-q^{-1} t^{-\ell(\lambda)+1}v/z \over 1-q^{-1}t v/z}
\prod_{i=1}^{\ell(\lambda)}
{1-q^{\lambda_i-1}t^{-i+2}v/z \over 1-q^{\lambda_i-1} t^{-i+1}v/z}
:\xi(q^{1/2}t^{-1/2}z) \eta_\lambda(v):,\\
&
 \eta_\lambda(v)\xi(q^{1/2}t^{-1/2}z)
={1-q t^{\ell(\lambda)-1}z/v \over 1-qt^{-1} z/v}
\prod_{i=1}^{\ell(\lambda)}
{1-q^{-\lambda_i+1}t^{i-2}z/v \over 1-q^{-\lambda_i+1} t^{i-1}z/v}
:\xi(q^{1/2}t^{-1/2}z) \eta_\lambda(v):.
\end{align*}
Then (\ref{xi-phi-1}), (\ref{xi-phi-2}) follow from (\ref{Phi0-3}), (\ref{Phi0-4}) in Lemma \ref{Phi0}.
\end{proof}

\begin{prp}\label{prop-3}
We have
\begin{align}
&\xi(q^{1/2}t^{-1/2}z)\tPhi_\lambda(v)+
q^{-1}t {v\over z} \tPhi_\lambda(v)\xi(q^{1/2}t^{-1/2}z)\nonumber\\
&\qquad =
\sum_{i=1}^{\ell(\lambda)} A^-_{\lambda,i} \, 
\tPhi_{\lambda-{\bf 1}_i}(v) \delta(q^{\lambda_i-1}t^{-i+1}v/z)
\varphi^+(q^{1/4}t^{-1/4}z).
\end{align}
\end{prp}
\begin{proof}
It follows from Lemmas \ref{delta-2}, \ref{xi-phi} and
$:\xi(q^{1/2}t^{-1/2}z)\eta(z):\,=\varphi^+(q^{1/4}t^{-1/4}z)$.

\end{proof}

\subsection{Operator product formulas for $\tPhis_\lambda(u)$}

\begin{lem}\label{Phis0}
We have
\begin{align}
&
\eta(z) \tPhis_\emptyset (u)=
(1-q^{-1/2}t^{1/2}u/z) 
:\eta(z) \tPhis_\emptyset (u):,\label{Phis0-1}\\
&
\tPhis_\emptyset (u)\eta(z) =
 (1-q^{1/2}t^{-1/2}z/u) 
  :\eta(z) \tPhis_\emptyset (u):,\label{Phis0-2}\\
&
\xi(z) \tPhis_\emptyset (u)={1\over 1-q^{-1}t u/z}:\xi(z) \tPhis_\emptyset (u):,\label{Phis0-3}\\
&
\tPhis_\emptyset (u)\xi(z) ={1\over 1-z/u} :\xi(z) \tPhis_\emptyset (u):,\label{Phis0-4}\\
&
\varphi^+(q^{-1/4}t^{1/4}z)^{-1}\tPhis_\emptyset (u)
={1-q^{-1}t u/z \over 1-u/z}  \tPhis_\emptyset (u)\varphi^+(q^{-1/4}t^{1/4}z)^{-1} ,\label{Phis0-5}\\
&
\varphi^-(q^{1/4}t^{-1/4}z)^{-1}\tPhis_\emptyset (u)
={1-qt^{-1} z/u \over 1-z/u}  \tPhis_\emptyset (u)\varphi^-(q^{1/4}t^{-1/4}z)^{-1} .\label{Phis0-6}
\end{align}
\end{lem}

\begin{prp}\label{prop-1s}
We have
\begin{align}
&
\varphi^+(q^{-1/4}t^{1/4}z)^{-1}\tPhis_\lambda (u)
\varphi^+(q^{-1/4}t^{1/4}z)
=B^+_\lambda(u/z) \tPhis_\emptyset (u) ,\label{prop-1s-1}\\
&
\varphi^-(q^{1/4}t^{-1/4}z)^{-1}\tPhis_\lambda (u)
\varphi^-(q^{1/4}t^{-1/4}z)
=B^-_\lambda(z/u)  \tPhis_\emptyset (u) .\label{prop-1s-2}
\end{align}

\end{prp}

\begin{proof}
Note that
\begin{align*}
&\varphi^+(q^{-1/4}t^{1/4}z)^{-1} \xi(u) \varphi^+(q^{-1/4}t^{1/4}z)
=g(u/z)^{-1}
\xi(u) ,\\
&
\varphi^-(q^{1/4}t^{-1/4}z)^{-1} \xi(u) \varphi^-(q^{1/4}t^{-1/4}z)
=g(z/u)
\xi(u) .
\end{align*}
Then (\ref{prop-1s-1}), (\ref{prop-1s-2}) follow from Lemmas  \ref{phi+-eta2} 
and (\ref{Phis0-5}), (\ref{Phis0-6}) in Lemma \ref{Phis0}.
\end{proof}


\begin{lem}\label{xi-phis}
We have
\begin{align}
&
\xi(z)\tPhis_\lambda(u)
=
{1\over 1-q^{-1} t^{-\ell(\lambda)+1} u/z} 
\prod_{i=1}^{\ell(\lambda)}
{1-q^{\lambda_i-1}t^{-i+1}u/z \over 1-q^{\lambda_i-1} t^{-i+2}u/z}
 :\xi(z)\tPhis_\lambda(u):,\label{xi-phis-1}\\
 &
B^+_\lambda(u/z)
\xi(z)\tPhis_\lambda(u)
=
{1\over 1- t^{-\ell(\lambda)} u/z} 
\prod_{i=1}^{\ell(\lambda)}
{1-q^{\lambda_i}t^{-i}u/z \over 1-q^{\lambda_i} t^{-i+1}u/z}
 :\xi(z)\tPhis_\lambda(u):,\label{xi-phis-2}\\
&
 \tPhis_\lambda(u)\xi(z)
=
{1\over 1-t^{\ell(\lambda)} z/u} 
\prod_{i=1}^{\ell(\lambda)}
{1-q^{-\lambda_i}t^{i}z/u \over 1-q^{-\lambda_i} t^{i-1}z/u}
 :\xi(z)\tPhis_\lambda(u):.\label{xi-phis-3}
\end{align}
\end{lem}
\begin{proof}
{}From
$\xi(z) \xi(u)=f(q^{-1}t u/z):\xi(z) \xi(u):$, and Lemma  \ref{eta-etalam2} we have
\begin{align*}
&
\xi(z) \xi_\lambda(u)=
{1-q^{-1} t u/z\over 1-q^{-1} t^{-\ell(\lambda)+1} u/z} 
\prod_{i=1}^{\ell(\lambda)}
{1-q^{\lambda_i-1}t^{-i+1}u/z \over 1-q^{\lambda_i-1} t^{-i+2}u/z}
:\xi(z) \xi_\lambda(u):,\\
&
 \xi_\lambda(u)\xi(z)=
 {1-z/u\over 1-t^{\ell(\lambda)} z/u} 
\prod_{i=1}^{\ell(\lambda)}
{1-q^{-\lambda_i}t^{i}z/u \over 1-q^{-\lambda_i} t^{i-1}z/u}
:\xi(z) \xi_\lambda(u):.
\end{align*}
Then (\ref{xi-phis-1}), (\ref{xi-phis-3}) follow from (\ref{Phis0-3}), (\ref{Phis0-4}) in  Lemma \ref{Phis0}.
\end{proof}

 \begin{prp}\label{prop-2s}
 We have
\begin{align}
 B^+_\lambda(u/z)
\xi(z)\tPhis_\lambda(u)+{z\over u}\tPhis_\lambda(u) \xi(z)
=
\sum_{i=1}^{\ell(\lambda)+1} A^+_{\lambda,i} \, 
\tPhis_{\lambda+{\bf 1}_i}(u) \delta(q^{\lambda_i}t^{-i+1}u/z).
\end{align}
\end{prp}
\begin{proof}
It follows from Lemma \ref{delta} and  (\ref{xi-phis-2}), (\ref{xi-phis-3})  in Lemma \ref{xi-phis}.
\end{proof}


\begin{lem}\label{eta-phis}
We have
\begin{align}
&
\eta(q^{1/2}t^{-1/2}z)  \tPhis_\lambda (u)
=(1-q^{-1}t^{-\ell(\lambda)+1}u/z)
\prod_{i=1}^{\ell(\lambda)} 
{1-q^{\lambda_i-1}t^{-i+2}u/z \over 1-q^{\lambda_i-1}t^{-i+1}u/z }
:\eta(q^{1/2}t^{-1/2}z)  \tPhis_\lambda (u):,\label{eta-phis-1}\\
&
\tPhis_\lambda (u)\eta(q^{1/2}t^{-1/2}z)  
=(1-q t^{\ell(\lambda)-1}z/u)
\prod_{i=1}^{\ell(\lambda)} 
{1-q^{-\lambda_i+1}t^{i-2}z/u \over 1-q^{-\lambda_i+1}t^{i-1}z/u }
:\tPhis_\lambda (u)\eta(q^{1/2}t^{-1/2}z)  :.\label{eta-phis-2}
\end{align}
\end{lem} 
\begin{proof}
Note that
$\eta(q^{1/2}t^{-1/2}z) \xi(u)=f(q^{-1} t u/z)^{-1} :\eta(q^{1/2}t^{-1/2}z) \xi(u):$,
and 
$\xi(u)\eta(q^{1/2}t^{-1/2}z)=f(z/u)^{-1}:\eta(q^{1/2}t^{-1/2}z) \xi(u):$.
{}From Lemma \ref{eta-etalam2} we have
\begin{align*}
&
\eta(q^{1/2}t^{-1/2}z) \xi_\lambda(u)
={1-q^{-1} t^{-\ell(\lambda)+1}u/z \over 1-q^{-1}t u/z}
\prod_{i=1}^{\ell(\lambda)}
{1-q^{\lambda_i-1}t^{-i+2}u/z \over 1-q^{\lambda_i-1} t^{-i+1}u/z}
:\eta(q^{1/2}t^{-1/2}z) \xi_\lambda(u):,\\
&
 \xi_\lambda(u)\eta(q^{1/2}t^{-1/2}z)
={1-q t^{\ell(\lambda)-1}z/u \over 1-qt^{-1} z/u}
\prod_{i=1}^{\ell(\lambda)}
{1-q^{-\lambda_i+1}t^{i-2}z/u \over 1-q^{-\lambda_i+1} t^{i-1}z/u}
:\eta(q^{1/2}t^{-1/2}z) \xi_\lambda(u):.
\end{align*}
Then (\ref{eta-phis-1}), (\ref{eta-phis-2}) follow from (\ref{Phis0-1}), (\ref{Phis0-2}) 
in Lemma \ref{Phis0}.
\end{proof}

\begin{prp}\label{prop-3s}
We have
\begin{align}
&\eta(q^{1/2}t^{-1/2}z)\tPhis_\lambda(u)+
q^{-1}t {u\over z} \tPhis_\lambda(u)\eta(q^{1/2}t^{-1/2}z)\nonumber\\
&\qquad =
\sum_{i=1}^{\ell(\lambda)} A^-_{\lambda,i} \, 
\tPhis_{\lambda-{\bf 1}_i}(v) \delta(q^{\lambda_i-1}t^{-i+1}v/z)
\varphi^-(q^{1/4}t^{-1/4}z)
\end{align}
\end{prp}

\begin{proof}
It follows from Lemmas \ref{delta-2}, \ref{eta-phis} and
$:\eta(q^{1/2}t^{-1/2}z)\xi(z):\, =\varphi^-(q^{1/4}t^{-1/4}z)$.

\end{proof}

\subsection{Final step of proofs}

The intertwining relations in Lemma \ref{lemma-1} 
are rewritten in terms of $\eta,\xi,\varphi^\pm $
as
\begin{align}
&
\varphi^+(q^{1/4}t^{-1/4}z )
\Phi_\lambda \varphi^+(q^{1/4}t^{-1/4}z )^{-1}=B^+_\lambda(v/z)\Phi_\lambda,
\label{shiki-1}\\
&
\varphi^-(q^{-1/4}t^{1/4}z )
\Phi_\lambda \varphi^-(q^{-1/4}t^{1/4}z )^{-1}=B^-_\lambda(z/v)\Phi_\lambda,
\label{shiki-2}\\
&\eta(z) \Phi_\lambda
-
{uz\over w} B^-_\lambda(z/v)  \Phi_\lambda \eta(z)\label{shiki-3}\\
&\qquad =
\sum_{i=1}^{\ell(\lambda)+1}
w^{-1}(q^{1/2}t^{-1/2}q^{\lambda_i} t^{-i+1}v)^{N+1}
A^+_{\lambda,i}\, \delta(q^{\lambda_i} t^{-i+1}v/z) \Phi_{\lambda+{\bf 1}_i},\nonumber\\
&
\xi(q^{1/2}t^{-1/2} z) \Phi_\lambda
-q^{-1}t {w\over uz} \Phi_\lambda\xi(q^{1/2}t^{-1/2} z)\label{shiki-4}\\
&\qquad 
=
\sum_{i=1}^{\ell(\lambda)}
w(q^{1/2}t^{-1/2}q^{\lambda_i-1} t^{-i+1}v)^{-N-1}
A^-_{\lambda,i}\, \delta(q^{\lambda_i-1} t^{-i+1}v/z) \Phi_{\lambda-{\bf 1}_i}
\psi^+(q^{1/4}t^{-1/4}z).\nonumber
\end{align}

\noindent
{\it Proof of Theorem \ref{thm-1}.}
{}From (\ref{shiki-1}) and (\ref{shiki-2}), 
we must have that 
$\Phi_\lambda$ be proportional to $\tPhi_\lambda(v)$ by virtue of
 Proposition \ref{prop-1}. 
Write $\Phi_\lambda=t(\lambda,v,u,N)\tPhi_\lambda(v)$.
 Then in view of Propositions \ref{prop-2}, \ref{prop-3},
 we find that  (\ref{shiki-3}) and (\ref{shiki-4})
may hold only in the case $w=-vu$ and 
when $t(\lambda,v,u,N)$ is  given by (\ref{t(lam)}).  \hfill $\square$

The intertwining relations in Lemma \ref{lemma-1} 
are rewritten in terms of $\eta,\xi,\varphi^\pm $
as

\begin{align}
&
\varphi^+(q^{-1/4}t^{1/4}z )^{-1}\Phi^*_\lambda  \varphi^+(q^{-1/4}t^{1/4}z )
=
B^+_\lambda(u/z)\Phi^*_\lambda ,\label{shiki-1s}\\
&
 \varphi^-(q^{1/4}t^{-1/4}z )^{-1}\Phi^*_\lambda  \varphi^-(q^{1/4}t^{-1/4}z )=
B^-_\lambda(z/u)\Phi^*_\lambda ,\label{shiki-2s}\\
&
B^+(u/z) \xi(z) \Phi^*_\lambda 
-
{vz\over w} \Phi^*_\lambda \xi(z)\label{shiki-3s}\\
&\qquad =
\sum_{i=1}^{\ell(\lambda)+1}
q^{-1}v(q^{1/2}t^{-1/2} q^{\lambda_i}t^{-i+1}u )^{-N} 
A^+_{\lambda,i}\delta(q^{\lambda_i}t^{-i+1}u/z) 
\Phi^*_{\lambda+{\bf 1}_i},\nonumber\\
&
\eta(q^{1/2}t^{-1/2}z) \Phi^*_\lambda
-
q^{-1}t {w\over vz} \Phi^*_\lambda\eta(q^{1/2}t^{-1/2}z)\label{shiki-4s}\\
&\qquad 
=
\varphi^{-}(q^{1/4}t^{-1/4}z)
\sum_{i=1}^{\ell(\lambda)+1}
qv^{-1}(q^{1/2}t^{-1/2} q^{\lambda_i-1}t^{-i+1}u )^{N}
 A^-_{\lambda,i}\delta(q^{\lambda_i-1}t^{-i+1}u/z) 
\Phi^*_{\lambda-{\bf 1}_i}.\nonumber
\end{align}

\noindent
{\it Proof of Theorem \ref{thm-2}.}
{}From (\ref{shiki-1s}) and (\ref{shiki-2s}), 
we must have that 
$\Phi^*_\lambda$ be proportional to $\tPhis_\lambda(v)$ by virtue of
 Proposition \ref{prop-1s}. 
Write $\Phi^*_\lambda=t^*(\lambda,v,u,N)\tPhis_\lambda(v)$.
 Then in view of Propositions \ref{prop-2s}, \ref{prop-3s},
 we find that  (\ref{shiki-3s}) and (\ref{shiki-4s})
may hold only in the case $w=-vu$ and 
when $t^*(\lambda,v,u,N)$ is  given by (\ref{ts(lam)}).  \hfill $\square$

\section{Proof of Proposition \ref{propos-1}}\label{section-7}
\subsection{Some formulas concerning $N_{\lambda,\mu}(u)$.}

\begin{lem}\label{lemma-N}
We have
\begin{align}
&
 N_{\lambda,\mu}(u)=
 \prod_{i=1}^{\ell(\lambda)} \prod_{j=1}^{\ell(\lambda)}
(u q^{-\mu_i+\lambda_{j+1}} t^{i-j};q)_{\lambda_j-\lambda_{j+1}}
\cdot
\prod_{\alpha=1}^{\ell(\mu)} \prod_{\beta=1}^{\ell(\mu)}
(u q^{\lambda_{\alpha}-\mu_\beta} t^{-\alpha+\beta+1};q)_{\mu_{\beta}-\mu_{\beta+1}},\label{NN}\\
&
N_{\lambda,\mu}(q^{1/2} t^{-1/2}x)=N_{\mu,\lambda}(q^{1/2} t^{-1/2}x^{-1})
x^{|\lambda|+|\mu|}{f_\lambda\over f_\mu},\\
&
c_\lambda c_\lambda'=
  (-1)^{|\lambda|} q^{n(\lambda')+|\lambda|}t^{n(\lambda)}
  N_{\lambda,\lambda}(1) . \label{ccp}
\end{align}
\end{lem}

\begin{lem}\label{lem-GN}
Let $\varepsilon^\pm_\lambda$ be the algebra homomorphism in (\ref{epsilon-pm}). 
We have
\begin{align}
\exp\left(\sum_{n>0}{1\over n}{1-t^n\over 1-q^n} 
\left(\varepsilon^+_{\lambda} p_n\right) 
\left( \varepsilon^-_{\mu}p_n\right)  u^n \right)=
{\mathcal G}(u)^{-1} N_{\lambda,\mu}(u), \label{GN}
\end{align}
where ${\mathcal G}(u)$ being as in (\ref{calG}).
\end{lem}

\begin{proof}
Fix an integer $\ell$ such that
 $\ell\geq{\rm Max}(\ell(\lambda),\ell(\mu))$. 
We have
\begin{align*}
{\rm LHS}&=
\exp\left(\sum_{n>0}{1\over n}{1-t^n\over 1-q^n} u^n
\left(
{t^{-n}\over 1-t^{-n}}+\sum_{i=1}^\ell (q^{\lambda_i n}-1)t^{-in}
\right)
\left(
{t^n\over 1-t^n}+\sum_{j=1}^\ell (q^{-\mu_j n}-1)t^{jn}
\right)\right)\\
&
=
{\mathcal G}(u)^{-1} 
\exp\left(\sum_{n>0}{1\over n}{1-t^n\over 1-q^n} u^n
\left(
\sum_{i,j=1}^\ell (q^{\lambda_i }t^{-i})^n(q^{-\mu_j }t^{j})^n \right. \right.\\
&\quad\quad\quad\quad \left.\left.+
{t^{-(\ell+1)n}\over 1-t^{-n}}
\sum_{j=1}^\ell (q^{-\mu_j }t^{j})^n+
{t^{(\ell+1)n}\over 1-t^{n}}
\sum_{i=1}^\ell (q^{\lambda_i }t^{-i})^n
\right)\right)\\
&={\mathcal G}(u)^{-1} 
\prod_{i=1}^{\ell} \prod_{j=1}^{\ell}
{(uq^{-\mu_i +\lambda_j} t^{i-j+1};q)_\infty \over 
(uq^{-\mu_i +\lambda_j} t^{i-j};q)_\infty} \cdot 
\prod_{k=1}^\ell 
{(uq^{-\mu_k } t^{k-\ell};q)_\infty \over (uq^{\lambda_k } t^{-k+\ell+1};q)_\infty},
\end{align*}
were we have used the notation
\begin{align*}
(u;q)_\infty =\exp\left(-\sum_{n=1}^\infty {1\over 1-q^n} u^n\right)
\in \bbQ(q)[[u]].
\end{align*}
Note that 
$(u;q)_\infty/(q^n u;q)_\infty=(u;q)_n$ ($n=0,1,2,\ldots$),  and use (\ref{NN}), then we have (\ref{GN}). 
 \end{proof}

\begin{prp}\label{Phi-Phi}
We have the operator product formulas 
\begin{align}
&
:\tPhis_\emptyset(z) \xi_\lambda(z): 
:\tPhis_\emptyset(w) \xi_\mu(w): ={{\mathcal G}(w/z)\over N_{\mu,\lambda}( w/z)}
:\tPhis_\emptyset(z) \xi_\lambda(z)
\tPhis_\emptyset(w) \xi_\mu(w):,
\\
&
:\tPhi_\emptyset(z) \eta_\lambda(z): 
:\tPhi_\emptyset(w) \eta_\mu(w): =
{{\mathcal G}(qt^{-1}w/z)\over N_{\mu,\lambda}(qt^{-1}w/z)}
:\tPhi_\emptyset(z) \eta_\lambda(z) 
\tPhi_\emptyset(w) \eta_\mu(w):,\\
&
:\tPhis_\emptyset(z) \xi_\lambda(z): 
:\tPhi_\emptyset(w) \eta_\mu(w): =
{N_{\mu,\lambda}(q^{1/2}t^{-1/2}w/z)\over {\mathcal G}(q^{1/2}t^{-1/2}w/z)}
:\tPhis_\emptyset(z) \xi_\lambda(z) 
\tPhi_\emptyset(w) \eta_\mu(w):,\\
&
:\tPhi_\emptyset(z) \eta_\lambda(z): 
:\tPhis_\emptyset(w) \xi_\mu(w): =
{N_{\mu,\lambda}(q^{1/2}t^{-1/2}w/z)\over {\mathcal G}(q^{1/2}t^{-1/2}w/z)}
:\tPhi_\emptyset(z) \eta_\lambda(z) 
\tPhis_\emptyset(w) \xi_\mu(w):.
\end{align}
\end{prp}

\begin{proof}
These follow from (\ref{Phi-1}), (\ref{Phis-1}) and (\ref{GN}).
\end{proof}

\subsection{Proof of Proposition \ref{propos-1}}

Using Lemma \ref{lemma-N} and Proposition \ref{Phi-Phi}, we have 
\begin{align*}
\mbox{LHS of (\ref{pro-1})}
 =&
 \sum_{\lambda^{(1)},\cdots,\lambda^{(N_c)}}
 \prod_{k=1}^{N_c}
 {\left(q^{-1/2}t^{1/2} v_i u_i^{-1}w_i^{L-M-2i+2}\right)^{|\lambda^{(i)}|}
 \over N_{\lambda^{(i)},\lambda^{(i)}}(1)}
 f_{\lambda^{(i)}}^{L-M-2i+1}\nonumber\\
&
\times \bra{0}
:\tPhis_\emptyset(-w_{N_c}) \xi_{\lambda^{(N_c)}}(-w_{N_c}): \cdots 
:\tPhis_\emptyset(-w_1) \xi_{\lambda^{(1)}}(-w_1): \ket{0} \nonumber \\
&
\times
\bra{0}:\tPhi_\emptyset(-w_{N_c}) \eta_{\lambda^{(N_c)}}(-w_{N_c}):\cdots
:\tPhi_\emptyset(-w_1) \eta_{\lambda^{(1)}}(-w_1):\ket{0} \nonumber\\
&
=
 \sum_{\lambda^{(1)},\cdots,\lambda^{(N_c)}}
 \prod_{k=1}^{N_c}
 {\left(q^{-1/2}t^{1/2} v_i u_i^{-1}w_i^{L-M-2i+2}\right)^{|\lambda^{(i)}|}
 \over N_{\lambda^{(i)},\lambda^{(i)}}(1)}
 f_{\lambda^{(i)}}^{L-M-2i+1}\nonumber\\
 &\times 
 \prod_{1\leq i<j\leq N_c} { {\mathcal G}(w_i/w_j)\over N_{\lambda^{(i)},\lambda^{(j)}} (w_i/w_j)}
  { {\mathcal G}(q t^{-1}w_i/w_j)\over N_{\lambda^{(i)},\lambda^{(j)}} (qt^{-1}w_i/w_j)}.
\end{align*}
Simplifying the factors by using Lemma \ref{simplify} below, we have the result. \hfill $\square$

\begin{lem}\label{simplify}
We have
\begin{align}
&
\prod_{1\leq i<j\leq N}(q^{1/2} t^{-1/2})^{-|\lambda^{(i)}|-|\lambda^{(j)}|}=
\prod_{i=1}^N (q^{1/2} t^{-1/2})^{-(N-1)|\lambda^{(i)}|},\\
&
\prod_{1\leq i<j\leq N}w_i^{-|\lambda^{(i)}|}w_j^{|\lambda^{(j)}|}=
\prod_{i=1}^N w_i^{(-N+2i-1)|\lambda^{(i)}|},\\
&
\prod_{1\leq i<j\leq N}w_i^{-|\lambda^{(j)}|}w_j^{|\lambda^{(i)}|}=
\prod_{i=1}^N 
\left(w_1w_2\cdots w_N\right)^{|\lambda^{(i)}|}w_i^{-|\lambda^{(i)}|}
\left(w_1w_2\cdots w_{i-1}\right)^{-2|\lambda^{(i)}|}.
\end{align}
\end{lem}

{\bf Acknowledgments}. 
Research of BF is partially supported by RFBR initiative interdisciplinary project grant 09-02-12446-ofi-m, 
by RFBR-CNRS grant 09-02-93106, RFBR grants 08-01-00720-a, 
NSh-3472.2008.2 and 07-01-92214-CNRSL-a. 
Research of JS is supported by the Grant-in-Aid for Scientific Research 
C-20272536. The authors thank A. Belavin, H. Kanno, V. Pasquier and Y. Yamada 
for stimulating discussions.


\end{document}